\documentclass[12pt, draftclsnofoot]{IEEEtran}
 \onecolumn
\usepackage{amssymb}
\usepackage{amsmath}
\usepackage[dvips]{graphics}
\usepackage{graphicx}
\usepackage{psfrag}
\usepackage{rotating}
 \allowdisplaybreaks
 \usepackage{cite}
\usepackage{theorem}
\theoremheaderfont{\upshape\bfseries} \theoremstyle{plain}
\usepackage{algcompatible}
\usepackage{algorithm}
\usepackage{float}
\usepackage{mathtools}

\DeclarePairedDelimiter{\floor}{\lfloor}{\rfloor}
\newfloat{initalgo}{h}

\theorembodyfont{\rmfamily}

\theoremheaderfont{\itshape}
\newtheorem{remark}{Remark} 
\newtheorem{example}{Example}

\newtheorem{theorem}{Theorem}

\newtheorem{lemma}{Lemma}

\begin{document}
\bstctlcite{IEEEexample:BSTcontrol}
\title{Initialization Algorithms For Convolutional Network Coding}
\author{Maxim Lvov and Haim H.  Permuter
\thanks{The work of M. Lvov and H. Permuter was supported by the Israel Science Foundation (grant no. 684/11), by the ERC starting grant and by the Israeli Ministry of Defense.}%
\thanks{This paper will be presented in part at the 2014 International Symposium on Network Coding, Aalborg, Denmark.}%
\thanks{M. Lvov and H. H. Permuter are with the Department of Electrical and Computer Engineering, Ben-Gurion University of the Negev, 84105, Beer-Sheva, Israel (email: maxlvo55@gmail.com; haimp@bgu.ac.il).}}
\maketitle
\begin{abstract}
We present algorithms for initializing a convolutional network coding scheme in networks that may contain cycles.
An initialization process is needed if the network is unknown or if local encoding kernels are chosen randomly.
During the initialization process every source node transmits basis vectors and every sink node measures the impulse response of the network. The impulse response is then used to find a relationship between the transmitted and the received symbols, which is needed for a decoding algorithm and to find the set of all achievable rates.
Unlike acyclic networks, for which it is enough to transmit basis vectors one after another, the initialization of cyclic networks is more involved, as  pilot symbols interfere with each other and the impulse response is of infinite duration. 
\end{abstract}


\section{Introduction}
 Network coding is a technique that is used to increase a network's throughput. The idea behind this coding scheme is that the relay nodes transmit functions of the received symbols on their output links, rather than simply routing them. Ahlswede \emph{et al.}\cite{CITE_NetworkInformationFlow} showed that for a one source, multicast, acyclic network, the maximal network's throughput is equal to the minimum cut between the source and any sink node. They also showed that for some networks, the ordinary routing scheme cannot achieve the min-cut bound, although a network coding scheme can.
  For cyclic networks, a Convolution Network Coding (CNC) scheme was presented by  Li \emph{et al.}\cite{CITE_LinearNetworkCoding}, and the existence of an optimal CNC code (one that achieves the min-cut bound given in \cite{CITE_NetworkInformationFlow}) was proved by Koetter and M\'{e}dard \cite{CITE_Algebraic}. Since then, much work has been devoted to constructing codes for cyclic networks \cite{CITE_Algebraic,CITE_On-convolutional,CITE_Commutative_Algebra,CITE_Efficient_codes,CITE_Efficient_code_design}, but all these code-construction algorithms share one major drawback; they all need to know in advance the network topology. In particular, if the network is large, it might be difficult to learn the exact network structure.
     
A randomized linear network coding approach was presented by Ho \emph{et al}\cite{CITE_Random_Multicast}. 
They showed that for a cyclic network, all sink nodes will be able, with high probability, to decode the symbols sent by the source nodes, provided that the transmission rates of all sources satisfy the Min-Cut Max-Flow condition and that the local encoding kernels are chosen randomly from a large enough field.
 The Min-Cut Max-Flow condition states that for every subset $\mathcal{A}$ of source nodes,
   the sum of source rates $\sum_{s\in\mathcal{A}} R_s$   must be less than or equal to the minimum cut capacity between every sink node and $\mathcal{A}$.
   
   This result makes random linear encoding extremely useful when the network is dynamic and no central authority for assigning encoding kernels exists. The local encoding kernels can be chosen randomly from some large enough field and, with high probability, this will lead to a network that allows source nodes to transmit symbols at high rates, thereby enabling all sink nodes to decode the sent symbols. This outcome, however, requires that the source nodes know the capacity region and that the sink nodes know a decoding algorithm. If the network structure or the local encoding kernels are not known, an initialization process is needed.

In this paper, we present two initialization algorithms that find a decoding scheme for the sink nodes and one algorithm that finds the capacity region for the source nodes.
     The decoding scheme is found by sending pilot basis vectors and measuring 
         the impulse response of the network, a method analogous to the one given in \cite[p.~448]{CITE_Raymond} for acyclic networks.
    Although the impulse response of the network can be of infinite duration, our algorithms find a decoding scheme using  only the initial values of the impulse response. 
    In the first algorithm, we transmit basis vectors and measure the impulse response of the network under the assumption that the initial symbols sent on the network are zeros. In the second algorithm, we assume that neither the initial symbols are zeros nor that it is possible to clear all these symbols at once. 
    Our algorithms do not require any additional headers to be transmitted. This simplifies the design of the relay nodes, since they do not operate differently during and after the initialization process. The method for finding the capacity region is based on the fact that the connection between the source and the sink nodes is possible if the transfer matrix is of full rank\cite{CITE_Algebraic}.
    
         A randomized initialization of convolutional network codes was introduced by Guo \emph{et al}\cite{CITE_Localized_dimension_growth}. Their method used a time-variant decoding algorithm proposed in \cite{CITE_time-variant_decoding} to decode the transmitted symbols. By that method, one can decode all of the transmitted symbols  up to time $n$ by using the first $n+L$ terms of the network's impulse response, where $L$ is the decoding delay. Our results can improve their algorithms since we have developed a method to find the full impulse response (by expanding the global encoding kernels, found in the initialization process, into power series) from a finite set of its  initial values. After finding the global encoding kernels, both time-variant decoding \cite{CITE_time-variant_decoding} presented by Guo \emph{et al}, and the sequential decoding algorithm\cite{CITE_Efficient_code_design} presented by Erez and Feder can be used.
    
        Methods for identifying an unknown linear time-invariant (LTI) system from its impulse response are well known from control theory. In particular, these methods are used to find a state space representation of the system, i.e. to find the matrices $\textbf{A},\textbf{B},\textbf{C}$ and $\textbf{D}$ such that the following state equations will satisfy the input-output relationship of the system:
        \begin{align}
        \textbf{x}[n+1]&=\textbf{A}\textbf{x}[n]+\textbf{B}\textbf{u}[n], \nonumber \\
            \textbf{y}[n]&=\textbf{C} \textbf{x}[n]+\textbf{D} \textbf{u}[n], \label{Intro_state_space}  
        \end{align}
        where $\textbf{u}[n]$ is the input vector, $\textbf{y}[n]$ is the output vector and $\textbf{x}[n]$ is the state vector.  Usually the state space representation obtained by these methods is an approximate one and is based on statistical methods\cite{CITE_sys_iden}. After such a representation is found, the transfer function can be found by applying the $Z$-transform on (\ref{Intro_state_space}):
    \begin{align}
    \textbf{H}(z)&=\textbf{C}\cdot adj(z\textbf{I}-\textbf{A}) \cdot \textbf{B}/P_\textbf{A}(z) + \textbf{D}. \label{Intro_transfer_function}
        \end{align}
    Here, $P_\textbf{A}(z)=\det(z\textbf{I}-\textbf{A})$ is the characteristic polynomial of $\textbf{A}$ and $ \textbf{H}(z)$ is the transfer function of the system. A cyclic network with a convolutional network coding scheme can also be described by a state space representation, as was introduced by Fragouli and Soljanin\cite{CITE_convolutional}. 
    
    System identification is closely related to the initialization process we show here. In both cases we have an unknown LTI system, for which we want to find the input-output relationship without learning the exact structure of the system, but only by sending pilot input vectors.
    However, our motivation for finding this input-output relationship differs from that usually cited in control theory, where we tend to look for the input sequence in order to obtain the desired  output sequence. In our case, we need the input-output relationship to be able to decode the transmitted symbols and to find the capacity region for all source nodes. There are also other differences between system identification and our initialization process, such as the fact that LTI systems usually work in the field of real or complex numbers, while the networks we work with use finite fields.
    
    One of the methods to find an input-output relationship of a deterministic LTI system from its impulse response is described by the Ho-Kalmans Method \cite[p.~142]{CITE_sys_iden}. Using that method, we first need to measure the impulse response $\left\{\textbf{G}[n]\right\}_{n=1}^{k+l}$ (where $k,l$ are any numbers that are greater than the order of the system) and construct the Hankel matrix
    \begin{align}
       \textbf{H}_{k,l}&=
       \begin{bmatrix}
       \textbf{G}_1 & \textbf{G}_2 & \textbf{G}_3 & \cdots & \textbf{G}_l \\[0.3em]
           \textbf{G}_2 & \textbf{G}_3 & \textbf{G}_4 & \cdots & \textbf{G}_{l+1} \\[0.3em]
           \textbf{G}_3 & \textbf{G}_4 & \textbf{G}_5 & \cdots & \textbf{G}_{l+2} \\[0.3em]
       \vdots & \vdots & \vdots & \ddots & \vdots \\[0.3em]
               \textbf{G}_{k} & \textbf{G}_{k+1} & \textbf{G}_{k+2} & \cdots & \textbf{G}_{l+k}
       \end{bmatrix}. \label{Hankel_Mat} 
       \end{align}
     Using a singular value decomposition (SVD) of $\textbf{H}_{k,l}$, a minimal realization of the system is constructed as described in \cite{CITE_sys_iden}, from which a transfer function is found. 
     This method, however, assumes that the field over which linear combinations are performed is the set of real or complex numbers. In our case, the field is finite and no SVD operation is defined.
     A method for system realization from its Hankel matrix is described in \cite[p.~498]{CITE_Rugh}, but it requires us to know  the rank of Hankel matrices of higher orders  (for larger $k,l$).

     The transfer function of an LTI system can be  found if one knows the characteristic polynomial $P_\textbf{A}(z)$ of the matrix $\textbf{A}$.  One can pass the output of the system through a finite impulse response (FIR) filter with a transfer function $P_\textbf{A}(z)$ such that the total transfer function of the cascaded system would be
          \begin{align}
         \textbf{H}(z) P_\textbf{A}(z)=\textbf{C}\cdot adj\left(z\textbf{I}-\textbf{A}\right)\cdot \textbf{B}. \label{Eq_IIR_to_FIR}
          \end{align}
          The transfer function in (\ref{Eq_IIR_to_FIR}) is a polynomial in $z$ and, hence, can be obtained by sending basis vectors and measuring the finite impulse response.
      A method to find the characteristic polynomial $P_\textbf{A}(z)$ from the diagonal minors of the Hankel matrix was introduced by Sreeram in \cite{CITE_HankelMatrix}. However, this method requires us to know the order of the system, i.e. the dimension of the state vector in its minimal realization, which is not usually known a priori when we consider unknown networks. In the methods we present only the number of edges and the maximal transmission rate for every source (or an upper bound for each of them) are needed.
      
  The paper is divided into seven sections. In Section II we outline notations and define the problem. In Section III we present two algorithms for network initialization and one for finding the capacity region of the network. In Sections IV, V and VI  we explain why these algorithms work, one algorithm per section. Section VII concludes the paper.
  In Appendix A we show examples for applying the algorithms. In Appendix B we give the proofs for all the theorems and lemmas.

\section{Notations and Problem Definition}

We represent a communication network by a directed graph $G=(\mathcal{V},\mathcal{E})$ where $\mathcal{V}$ is the set of nodes and $\mathcal{E}$ is the set of edges. Each edge represents a noiseless directed link that can transmit one symbol per unit time, where the symbols are scalars from some field $\mathbb{F}$. We assume every link has a unit time delay between consequent symbol transmissions and transmissions on all links are synchronized.

We denote by $\mathcal{S}$ the set of all source nodes and by $\mathcal{D}$ the set of all sink nodes. Every source node $s\in\mathcal{S}$ transmits $R_s$ symbols per unit time. Every sink node wants to receive all the symbols sent by all the source nodes. For every edge $e\in\mathcal{E}$, we say that $u=head(e)$ and $v=tail(e)$ if $u,v\in\mathcal{V}$ and $e$ is from $v$ to $u$. We  denote by $In(u)=\left\{e\in\mathcal{E} : u=head(e)\right\}$ and $Out(u)=\left\{e\in\mathcal{E} : u=tail(e)\right\}$. The symbol that is sent on the edge $e$ at time $n\in\mathbb{Z}$ is denoted by $x_e[n]$.
We denote vectors or sequences of vectors by lowercase  bold letters, while matrices are denoted by bold capital letters.
We assume there is a CNC scheme in the network, so that the symbol sent on a link $i\in Out(j)$ is a linear combination of the symbols received and generated by the node $j$ in the previous time slot. This relationship can be written as
\begin{align}
x_i[n+1]&=\sum_{e\in In(j)} a_{i,e} x_e[n] + \sum_{k=1}^{R_j} b_{i,k}u_{j,k}[n], \hspace{3 mm} \forall i\in\mathcal{E}, \hspace{1 mm} \forall n\geq 0, \label{LNC}
\end{align}
where $u_{j,k}[n]$ is the $k$'th symbol generated by node $j$ (if $j\in\mathcal{S}$) at time $n$, and $\left\{a_{i,e}, b_{i,k}\right\}$ are the \textit{local encoding kernels} for node $j$ that were chosen in advance (probably randomly). By letting $x_i[n]$ depend only on the previously received symbols, we avoid the problem described in \cite{CITE_Conditions_Determine_CNC} by Cai and Guo, when the convolutional code is not well defined in a cyclic network. 
If the network has a reset option that clears all the sent symbols in the network, we can assume that the initial network state is zero:
\begin{align}
x_i[0]=0, \hspace{5 mm} \forall i\in\mathcal{E}. \nonumber
\end{align}
\begin{example}
As an example, we consider the network in Fig.\ref{fig_example1}. There is one source node $s_1$ and one sink node $d_1$. By the Min-Cut Max-Flow Theorem the rate $R_{s_1}=1$ is achievable, and the network state equations can be written in the next form: \label{example_1}
\begin{align}
	\begin{pmatrix}
       x_1[n+1] \\
       x_2[n+1] \\
       x_3[n+1] \\
       x_4[n+1]
      \end{pmatrix} =
    \begin{pmatrix}
     0 & 0 & 0 & \alpha_{1,4}  \\
     \alpha_{2,1} & 0 & \alpha_{2,3} & 0  \\
     0 & 0 & 0 & \alpha_{3,4} \\
	 0 & \alpha_{4,2} & 0 & 0
    \end{pmatrix}
        \begin{pmatrix}
               x_1[n] \\
               x_2[n] \\
               x_3[n] \\
               x_4[n]
              \end{pmatrix} + \textbf{B}_{s_1}\cdot\textbf{u}_{s_1}[n].
    \end{align}
\begin{figure}[t!]{
\psfrag{s1}[][][1]{$s_1$}
\psfrag{d1}[][][1]{$d_1$}
\psfrag{r1}[][][1]{$ $}
\psfrag{x1}[][][1]{$x_1[n]$}
\psfrag{x2}[][][1]{$x_2[n]$}
\psfrag{x3}[][][1]{$x_3[n]$}
\psfrag{x4}[][][1]{$x_4[n]$}
 \centerline{\includegraphics[width=7cm]{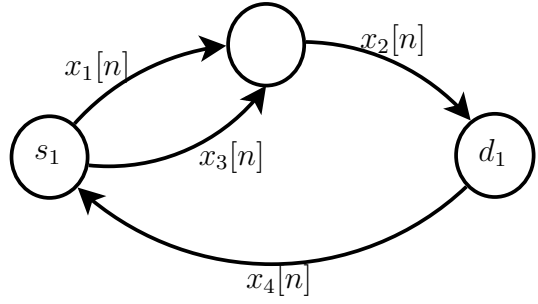}}
\caption{Network with one source node, one sink node and one relay node.} \label{fig_example1}
}\end{figure}
If the rate $R_{s_1}$ is set to $1$ then ${u}_{s_1}[n]$ is a scalar sequence and $\textbf{B}_{s_1}=(b_{1,1}, 0, b_{3,1}, 0)^T$ is a $4\times 1$ matrix over $\mathbb{F}$. 
\end{example}
Note that we have restricted ourselves to the case where local encoding kernels are scalars, while in the general case they can be rational power series in the time shift operator \cite[p.~492]{CITE_Raymond}. This, however, is not a major restriction, since one can achieve the capacity region without rational local encoding kernels if the field one works with is large enough \cite[p.~502]{CITE_Raymond}. Nevertheless, we treat separately network codes with rational power series encoding kernels at the end of section III. 

We define a time shift operator $z$ acting on a sequence (of scalars or vectors) $\{c[n]\}_{n\in\mathbb{Z}}$ as follows:
\begin{align}
(z^kc)[n]=c[n+k], \forall k,n\in \mathbb{Z}. \label{timeshift}
\end{align}
Let $P(t)=\sum_{k=0}^M a_kt^k$ be a polynomial in $t$ with coefficients from the field $\mathbb{F}$. We define the operator $P(z)$ as follows:
\begin{align}
\left(P(z)c\right)[n]=\sum_{k=0}^M a_kc[n+k], \hspace{1 mm} \forall k,n\in \mathbb{Z}. \label{Polynomial_operator}
\end{align}
Finally, the coefficients $\{a_k\}$ of $P(t)$ can also be $m\times k$ matrices over the field $\mathbb{F}$. In that case, the sequence $\{\textbf{c}[n]\}_{n\in\mathbb{Z}}$ in (\ref{Polynomial_operator}) should be a sequence of $k\times 1$ vectors. In order to avoid ambiguity, we will not use Z transforms of sequences and the symbol $z$ will appear only as a time shift operator.

  Let $\textbf{x}[n]$ be the column vector of size $|\mathcal{E}|$ consisting of all symbols $\{x_e[n]\}_{e\in\mathcal{E}}$ organized in some order. We define the \textit{input sequence} $\textbf{u}[n]=\left(\textbf{u}_{s_1}^T[n],...,\textbf{u}_{s_{|\mathcal{S}|}}^T[n]\right)^T$ where $\textbf{u}_{s_i}[n]=\left(u_{s_i,1}[n],...,u_{s_i,R_{s_i}}[n]\right)^T$ is the \textit{input sequence  of source $s_i$}, which is a sequence of vectors sent by source $s_i$. The dimension of the column vector $\textbf{u}[n]$ is $m=\sum_{s\in\mathcal{S}}R_s$. We assume that $\textbf{u}[n]=\textbf{0}$ for $n<0$.
  For every sink node $d$, we let $\textbf{y}_d[n]$ be a column vector consisting of all received symbols $\{x_e[n] : e\in In(d)\}$ and the symbols generated by $d$,$\{u_{d,k}[n]\}_{k=1}^{R_d}$ if $d$ is also a source node, again organized in some order. The sequence $\left\{\textbf{y}_d[n]\right\}_{n\in\mathbb{Z}}$ will be called the $\textit{output sequence}$ of the sink node $d$, and the dimension of every vector in that sequence is $l_d=R_d+\left\vert{In(d)}\right\vert$. We assume also that $\textbf{y}_d[n]=\textbf{0}$ for $n<0$.
  \begin{example}
  The shuttle network shown in Fig. \ref{fig_conversation_network} is used as an example.  The nodes $s_1,s_2$ are both source and sink nodes, and have the same transmission rates $R_{s_1}=R_{s_2}=1$. The state vector is $\textbf{x}[n]=\left(x_1[n],x_2[n],...,x_8[n]\right)^T$, the input sequence is $\textbf{u}[n]=\left(u_{s_1,1}[n], u_{s_2,1}[n]\right)^T$ ($m=2$), and the output sequences are $\textbf{y}_{s_1}[n]=\left(x_6[n],u_{s_1,1}[n]\right)^T$ and $\textbf{y}_{s_2}[n]=\left(x_7[n],u_{s_2,1}[n]\right)^T$. Both $l_{s_1}$ and $l_{s_2}$ are equal to 2.
  \begin{figure}[h]{
  \psfrag{s1}[][][1]{$s_1$}
  \psfrag{s2}[][][1]{$s_2$}
  \psfrag{x1}[][][1]{$x_1[n]$}
  \psfrag{x2}[][][1]{$x_2[n]$}
  \psfrag{x3}[][][1]{$x_3[n]$}
  \psfrag{x4}[][][1]{$x_4[n]$}
  \psfrag{x5}[][][1]{$x_5[n]$}
  \psfrag{x6}[][][1]{$x_6[n]$}
  \psfrag{x7}[][][1]{$x_7[n]$}
  \psfrag{x8}[][][1]{$x_8[n]$}
   \centerline{\includegraphics[width=12cm]{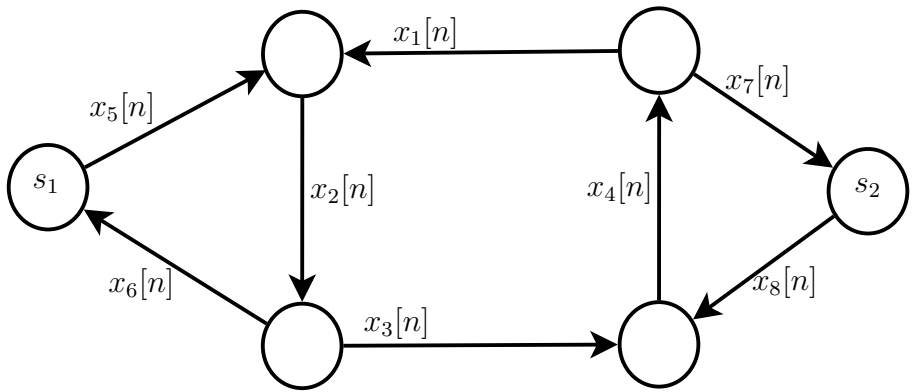}}
  \caption{Shuttle network with two users and 4 relay nodes.} \label{fig_conversation_network}
  }\end{figure}
  \end{example}

 We assume that either the network topology or the local encoding
  kernels or both are not known to any node a priori. We are interested in finding the network transfer matrix, or a way to decode the sent symbols $\{\textbf{u}[n]\}$ from the received symbols $\{\textbf{y}_d[n]\}$ at every sink node $d$, probably with some delay. This transfer function is obtained in our algorithms by sending pilot symbols and measuring the impulse response of the network. Even though we assume the network is unknown, our algorithms need all source and sink nodes to know some parameters of the network before the initialization process starts. These parameters can be shared by some distribution protocol or assumed to be known a priori. The parameters are:
\begin{itemize}
\item The set of source nodes $\mathcal{S}$ and their transmission rates $\{R_s\}_{s\in\mathcal{S}}$ (or an upper bound for every rate),
\item The number of edges in the network or an upper bound for it, which will be called $N$.
\end{itemize}
  
We next define the \textit{achievable rates} for the network with specific local encoding kernels as the transmission rates $\left(R_s\right)_{s\in\mathcal{S}}$ of all source nodes that will allow every sink node $d$ to decode the vectors $\left\{\textbf{u}[n]\right\}_{n=0}^{n_0}$ from the vectors $\left\{\textbf{y}_d[n]\right\}_{n=0}^{n_0+\delta_d}$ (where $\delta_d\geq 0$ represents the decoding delay, and is independent of $n_0$) for all $n_0\in\mathbb{N}$. The \textit{capacity region} is defined as the set of all achievable rates.

In the definitions above, we do not restrict the sink nodes to any decoding method, even if these methods use the knowledge of the network topology and the local encoding kernels at every node. We do, however, restrict the network to have a CNC scheme with the chosen local encoding kernels. This restriction is not of great importance, since a CNC scheme with randomly chosen local encoding kernels can reach the capacity region given by the Min-Cut Max-Flow condition.

Before the initialization process starts, a transmission rate for every source node should be chosen.
If an achievable rate for a specific source node is not known, it is preferable to set its rate to $R_s=|Out(s)|$. Algorithm 3, presented in the next section, can then be used to find achievable rates for this source. The source node $s$ can then reduce its rate $R_s$ to an achievable one by sending zeros on some of its input sequences $\left\{u_{s,1},...,u_{s,|Out(s)|}\right\}$.
In that case, we call the rates $\left(R'_s\right)_{s\in\mathcal{S}}$ \textit{achievable for a sink node $d$} if that sink node can decode the input sequence $\textbf{u}$ from the output $\textbf{y}_d$ when every source node $s$ transmits symbols on $R'_s$ out of its input sequences, and zeros on the rest of the $(R_s-R'_s)$ input sequences. Note that the rates $\left(R'_s\right)_{s\in\mathcal{S}}$ are achievable if they are achievable for every sink node.

\begin{example}
Recall the network from Fig. \ref{fig_example1} that was considered in Example \ref{example_1}. If the network topology and the capacity region are not known to $s_1$, the rate $R_{s_1}$ can be set to 2 as there are two outgoing links from the source node. In that case, $\textbf{u}_{s_1}[n]$ will be taken as a $2\times 1$ vector sequence and $\textbf{B}_{s_1}$ as a $4\times 2$ matrix:

\begin{align}
\textbf{B}_{s_1}= \begin{pmatrix}
               b_{1,1} & b_{1,2} \\
               0 & 0 \\
               b_{3,1} & b_{3,2} \\
               0 & 0
              \end{pmatrix}, \hspace{4 mm} 
\textbf{u}_{s_1}[n]= \begin{pmatrix}
               u_{s_1,1}[n] \\
               u_{s_1,2}[n] 
              \end{pmatrix}.           
\end{align}
After finding the capacity region, the rate can be reduced to an achievable one by sending zeros on one of the input sequences
\begin{align}
u_{s_1,1}[n]=0, \forall n\in\mathbb{N} \hspace{5 mm}\text{or} \hspace{5 mm} u_{s_1,2}[n]=0 , \forall n\in\mathbb{N}.
\end{align}
\end{example}

\section{The Initialization Algorithms}
In this section, we present two initialization algorithms that find a decoding scheme for the sink nodes and one algorithm that finds the capacity region for the source nodes. The purpose of the first two is to find a difference equation of the following form: 
   \begin{align}
   P_d(z)\textbf{y}_d=\textbf{G}_d(z)\textbf{u}. \label{diffeq}
   \end{align}
  This form describes the relationship between the transmitted sequence $\textbf{u}[n]$ and the received sequences $\textbf{y}_d[n]$ (for every sink node $d$).
 In (\ref{diffeq}), $P_d(z)$ is a polynomial in the time shift operator $z$, and $\textbf{G}_d(z)$ is a matrix with polynomial elements. These operators are defined in Section II.
  Using a decoding method similar to the one shown in \cite{CITE_Efficient_code_design}, we can show 
     that it is possible to decode the input sequence from the output when the polynomial $P_d(z)$ is not the zero polynomial and the \textit{transfer matrix}  $\textbf{G}_d(z)$ is of full column rank over the polynomial ring $\mathbb{F}[z]$.

   The difference between the two initialization algorithms is that in the first, it is assumed that we can perform a reset operation on the network at some fixed times and, therefore, this algorithm is a bit faster than the second algorithm that does not operate under this assumption. 
   The purpose of the third algorithm is to find achievable rates for all source nodes. 
   This is done by examining the transfer matrix $\textbf{G}_d(z)$ for every sink node $d$. To obtain this matrix, one of the initialization algorithms should be used first.

We now present the first algorithm.
Its first part consists of $\left(\sum_{s\in\mathcal{S}} R_s\right)$ loops. Every loop takes $2N+1$ time units, and after each loop the symbols on all edges are cleared.
Algorithm 1 is applied in Example \ref{example_alg1} in the appendix.

\begin{algorithm}[!htp]
  \caption{Initialization algorithm with network resetting}

  \begin{algorithmic}
   \State
   \Statex \begin{enumerate}
   \item For every $s\in\mathcal{S}$, and for every $j\in\{1,2,...,R_s\}$ do the following:
   \begin{itemize}
   \item Send the sequence $\textbf{u}_{s}^j[n]=\left(u_{s,1}^j[n],...,u_{s,R_s}^j[n]\right)^T$ at the times $n=0,1,...,2N$, where
   \begin{align}
   u_{s,i}^j[n]=\begin{cases}
       1,& i=j \hspace{5 mm} \text{and} \hspace{5 mm} n=0\\
       0,& i\neq j \hspace{5 mm} \text{or} \hspace{5 mm} 1\leq n \leq 2N
   \end{cases} \hspace{5 mm} \forall i\in\left\{1,...,R_s\right\}.
   \end{align}
   \item For all source nodes $\tilde{s}\neq s$, send zeros on their input sequences: $\textbf{u}_{\tilde{s}}^j[n]=\textbf{0}$.
  \item Every sink node $d$ should store its received vectors $\left\{\textbf{y}_{d}^{s,j}[n]\right\}_{n\in\{0,..,2N\}, s\in\mathcal{S}, j\in\{1,...,R_s\}}$, where each vector $\textbf{y}_{d}^{s,j}[n]$ is of dimension $l_d$.
   \item Reset the network after $n=2N$, by setting $n=0$ and $\textbf{x}[0]=\textbf{0}$.
   \end{itemize}
   \item For every sink node $d$ do the following:
      \begin{itemize}
      \item Combine the received vectors into matrices of size $l_d\times m$:
      \begin{align}
      \textbf{M}_d[n]=\left[\textbf{y}_{d}^{s_1,1}[n],...,\textbf{y}_{d}^{s_1,R_{s_1}}[n],\textbf{y}_{d}^{s_2,1}[n],...,\textbf{y}_{d}^{s_2,R_{s_2}}[n],...,\textbf{y}_{d}^{s_{|\mathcal{S}|},R_{s_{|\mathcal{S}|}}}[n]\right]. \label{matrix_combining}
      \end{align}
      \item Find any non trivial solution to the set of linear equations
         \begin{align}
         \sum_{k=0}^{N}\alpha_{d,k}\textbf{M}_d[k+\tau]=\textbf{O}, \forall \tau=1,...,N,
          \label{mamal}
         \end{align}
         where $\textbf{O}$ is the $l_d\times m$ zero matrix and $\{\alpha_{d,k}\}_{k=0}^{N}\subseteq \mathbb{F}$ are the unknowns. This set has $l_d\times m \times N$ equations and it has always a non trivial solution.
     
      \end{itemize}

	\algstore{myalg1}

    \end{enumerate}
     \end{algorithmic}
   \end{algorithm}
   \begin{algorithm}
   \begin{algorithmic} [1]
   \algrestore{myalg1}
   \STATE

   \begin{itemize}
   \item Construct the polynomial $P_d(z)$ and the matrix $\textbf{G}_d(z)$ as
           \begin{align}
           P_d(z)&=\sum_{k=0}^{N} \alpha_{d,k}z^k, \label{Polynom} \\
           \textbf{G}_d(z)&=\sum_{k=1}^{N}\sum_{j=k}^N \alpha_{d,j} \textbf{M}_d[j-k+1]z^{k-1}+\textbf{M}_d[0]P_d(z). \label{TransMatrix}
           \end{align}
           \item The difference equation that describes the relationship between the input and the output sequences $\textbf{u}[n]$ and $\textbf{y}_d[n]$  is given in (\ref{diffeq}), with the polynomial $P_d(z)$ and the matrix $\textbf{G}_d(z)$ as defined in (\ref{Polynom}-\ref{TransMatrix}). If $\textbf{G}_d(z)$ is of full column rank over the polynomial ring $\mathbb{F}[z]$, then $\textbf{u}[n]$ can be decoded from $\textbf{y}_d[n]$ by solving (\ref{diffeq}). Otherwise, the transmission rates $\left\{R_s\right\}$ of some source nodes should be reduced, or other local encoding kernels should be chosen.

   \end{itemize}

   \end{algorithmic}
  \end{algorithm}
\newpage
We now present the second algorithm, in which no resetting operation is needed. We consider the case when the network initial state is $\textbf{x}_0\neq\textbf{0}$ and $\textbf{x}_0$ is unknown. Algorithm 2 is similar to the first, except that this algorithm takes additional $2N+1$ time units (only in case $\textbf{x}_0\neq\textbf{0}$) and the expression for obtaining $\textbf{G}_d(z)$ is a bit different. If $\textbf{x}_0=\textbf{0}$ then we can skip the operations in the first $(2N+1)$ time units since the measured output vectors will contain only zeros. Algorithm 2 is applied in Example \ref{example_alg2} in the appendix.

\begin{algorithm}[!htp]
  \caption{Initialization algorithm without network resetting}

  \begin{algorithmic}
   \State
   \Statex \begin{enumerate}
   \item The input sequence $\textbf{u}[n]=\left(u_1[n],...,u_m[n]\right)^T$ that should be sent is
   \begin{align}
   u_i[n]=\begin{cases}
          1,& n=(2N+1)i \\
          0,& \text{otherwise}
      \end{cases} \hspace{5 mm} \forall i\in\left\{1,...,m\right\},\hspace{2 mm} 0\leq n \textless (m+1)(2N+1), \label{pilot_sequence}
   \end{align}
   where $m$ is the dimension of $\textbf{u}[n]$. Note that to send the above sequence, every source node $s\in\mathcal{S}$ should send the symbol $1$ on every one of its inputs in turn $(u_{s,1},...,u_{s,R_s})$ at the correct time, and zeros at all other times.

	\algstore{myalg2}

    \end{enumerate}
     \end{algorithmic}
   \end{algorithm}
 \begin{algorithm}
   \begin{algorithmic} [1]
   \algrestore{myalg2}
   \STATE
   \begin{enumerate}
   \setcounter{enumi}{1}
\item For every sink node $d$ do the following:
  
   \begin{itemize}
         \item Find any non trivial solution to the set of linear equations
         \begin{align}
         	&\sum_{j=0}^N \alpha_{d,j} \textbf{y}_d[j+\tau]=\textbf{0}, \hspace{5 mm} \forall \tau\in \bigcup_{p=0}^{m}\bigcup_{\tilde{\tau}=1}^N\{(2N+1)p+\tilde{\tau}\}, \label{mamalNew}
            \end{align}
             where $\{\alpha_{d,j}\}_{j=0}^{N}\subseteq \mathbb{F}$ are the unknowns. This set has $l_d \times N \times (m+1)$ equations, and it
             always has a non trivial solution.
   \item The polynomial $P_d(z)$ and the matrix $\textbf{G}_d(z)$ are defined below:
                \begin{align}
                P_d(z)&=\sum_{k=0}^{N} \alpha_{d,k}z^k \label{Polynom2}, \\
                \textbf{g}_{d,i}(z)&=\sum_{k=1}^{N+1} \sum_{j=0}^N \alpha_{d,j}\textbf{y}_{d}[j+(2N+1)i-k+1]z^{k-1}, \hspace{5 mm} \forall i\in\{1,...,m\}, \label{TransMatrixNew2} \\
               \textbf{G}_d(z)&=\left[\textbf{g}_{d,1}(z),\textbf{g}_{d,2}(z),...,\textbf{g}_{d,m}(z)\right]. \label{TransMatrixNew1} 
                \end{align}
   \item The difference equation that describes the relationship between the input and the output sequences $\textbf{u}[n]$ and $\textbf{y}_d[n]$ for $n\geq 1$ is given in (\ref{diffeq}), with the polynomial $P_d(z)$ and the matrix $\textbf{G}_d(z)$ as defined in (\ref{Polynom2}-\ref{TransMatrixNew1}). If $\textbf{G}_d(z)$ is of full column rank over the polynomial ring $\mathbb{F}[z]$, then $\textbf{u}[n]$ can be decoded from $\textbf{y}_d[n]$ by solving (\ref{diffeq}). Otherwise, the transmission rates $\left\{R_s\right\}$ of some source nodes should be reduced, or other local encoding kernels should be chosen.

   \end{itemize}
   \end{enumerate}	

   \end{algorithmic}
  \end{algorithm}
We now present the third algorithm that allows us to find achievable rates for all source nodes in the network, with the chosen local encoding kernels. It uses the matrix $\textbf{G}_d(z)$ from (\ref{diffeq}) and hence Algorithm 1 or 2 should be used first to find the matrix. At the end of this algorithm, every sink node $d$ will be able to tell what rates are achievable for it. 

\begin{algorithm}[!h]
  \caption{Finding the capacity region}

  \begin{algorithmic}
   \State
   \Statex
   	The capacity region is found as follows:
   	\begin{itemize}
      	\item For every sink node $d$, split the matrix $\textbf{G}_d(z)$ into $|\mathcal{S}|$ matrices, such that each matrix $\textbf{G}_{d,s}(z)$ has $R_s$ columns and such that the following will hold:
         \begin{align}
         	\textbf{G}_d(z)\textbf{u}&=\left[
         	\textbf{G}_{d,s_1}(z),...,\textbf{G}_{d,s_{|\mathcal{S}|}}(z) \right]
         	\begin{bmatrix}
         	\textbf{u}_{s_1} \\ \vdots \\ \textbf{u}_{s_{|\mathcal{S}|}} 
         	\end{bmatrix} \nonumber \\
         	&=\sum_{s\in \mathcal{S}} \textbf{G}_{d,s}(z)\textbf{u}_s.
         	\end{align}

      \item For every possible $n$-tuple $(R'_s)_{s\in\mathcal{S}}$ with integer
      entries that satisfy $R'_s\leq R_s$, check if for every source node $s$, there exist $R'_s$ column vectors $\left\{\textbf{v}_{s,1},...,\textbf{v}_{s,R_s'}\right\}$ in the columns of the matrix $\textbf{G}_{d,s}(z)$  such that all the vectors $\cup_{s\in\mathcal{S}}\cup_{k=1}^{R_s'}\{\textbf{v}_{s,k}\}$ are linearly independent over the polynomial ring $\mathbb{F}[z]$. If there are such vectors, the rates $(R'_s)_{s\in\mathcal{S}}$ are achievable for the sink node $d$.
      \item The capacity region is obtained by taking all $n$-tuples $(R'_s)_{s\in\mathcal{S}}$ that are achievable for every sink node.
      \end{itemize}
     \end{algorithmic}
   \end{algorithm}

Algorithm 3 allows us to find achievable rates with the currently chosen encoding kernels. If they were chosen randomly from a large enough field, these rates will be, with high probability, all the rates from the capacity region. There is, however, a small probability that the local encoding kernels were not chosen well. In that case, Algorithm 3 will only give the achievable rates with the currently chosen coefficients. Algorithm 3 is applied in Example \ref{example_alg3} in the appendix.
\begin{remark}
Although we restricted ourselves to the case of scalar local encoding kernels, the algorithms can be extended to networks that use CNC with rational power series as local encoding kernels\cite[p.~492]{CITE_Raymond}. In this case the input-output relationship of each node $j\in\mathcal{V}$ can be described by state space equations\cite[p.~481]{CITE_Rugh}.  Denote the state vector of node $j$ by $\tilde{\textbf{x}}_j[n]$, and its dimension by $\dim{\tilde{\textbf{x}}_j[n]}$.
If we concatenate all state vectors $\{\tilde{\textbf{x}}_j[n]\}_{j\in\mathcal{V}}$ into one state vector $\tilde{\textbf{x}}[n]$ of dimension $\sum_{j\in\mathcal{V}} \dim(\tilde{\textbf{x}}_j[n])$, a global state space representation of the network can be written:
\begin{align}
\tilde{\textbf{x}}[n+1]&=\hat{\textbf{A}}\tilde{\textbf{x}}[n]+\hat{\textbf{B}}{\textbf{u}}[n], \label{global_state_space1} \\
{\textbf{y}}_d[n]&=\hat{\textbf{C}}_d \tilde{\textbf{x}}[n]+\textbf{D}_d\textbf{u}[n], \label{global_state_space2}  
\end{align}
where $\hat{\textbf{A}}$, $\hat{\textbf{B}}$, $\hat{\textbf{C}}_d$ and $\textbf{D}_d$ are defined by the network topology and the local encoding kernels. 
        The derivation of our algorithms is based only on the fact that the input-output relationship of the network can be written as state space equations with a state vector of dimension $|\mathcal{E}|\leq N$. In the case where we use rational power series as local encoding kernels,  the algorithms will still apply if we take $N$ to be larger that the dimension of the state vector $\tilde{\textbf{x}}$:
         \begin{align}
         N\geq \sum_{j\in\mathcal{V}} \dim(\tilde{\textbf{x}}_j[n]).
         \end{align}
\end{remark}
\section{Derivation of Algorithm 1}
Our goal is to find a relationship between the input sequence $\textbf{u}[n]$ and the output sequence $\textbf{y}_d[n]$ for every sink node $d$ that will allow it to decode the sent symbols. Such a relationship can be given in the form of a difference equation, similar to that given in (\ref{diffeq}). The problem is how to find a polynomial $P_d(z)$ and a matrix $\textbf{G}_d(z)$ that will satisfy (\ref{diffeq}) for all $n\geq 0$ only from the received symbols $\textbf{y}_d[n]$. We assume, without loss of generality, that $N$ is equal to the number of edges in the network. However, if $N$ is larger, we can assume that there are an additional $2\left(N-|\mathcal{E}|\right)$ virtual nodes and $\left(N-|\mathcal{E}|\right)$ virtual edges between these nodes. The virtual edges are not connected to the original network and have no influence on it. In this way, the number of edges in the new network is $N$.
We observe that in view of (\ref{LNC}) and by the definition of $\textbf{x}[n], \textbf{u}[n]$ and $\textbf{y}_d[n]$, for every sink node $d$, a state space representation of the network can be written as
\begin{align}
\textbf{x}[n+1]&=\textbf{A}\textbf{x}[n]+\textbf{B}\textbf{u}[n], \hspace{3 mm} \textbf{x}[0]=\textbf{x}_0, \hspace{3 mm} \forall n\geq 0,  \label{state-space-1]} \\
\textbf{y}_d[n]&=\textbf{C}_d\textbf{x}[n]+\textbf{D}_d\textbf{u}[n], \hspace{3 mm} \forall n\in\mathbb{Z}, \label{state-space-2]}
\end{align}
where the matrices \textbf{A} and \textbf{B} are of sizes $N\times N$ and $N \times m$, respectively, and are determined by the network structure and the local encoding kernels on every node. An example of the matrices $\textbf{A}$ and $\textbf{B}$ is shown in Example \ref{example_1}. The matrices $\textbf{C}_d$ and $\textbf{D}_d$ contain only ones and zeros and are chosen so that $\textbf{y}_d[n]$ will contain the incoming symbols and the symbols generated by $d$, if $d$ is a source node. 

A general solution to the state equations (\ref{state-space-1]}) and (\ref{state-space-2]}) is given by
\begin{align}
\textbf{y}_d[n]=\textbf{C}_d\textbf{A}^{n}\textbf{x}_0  +\sum_{i=0}^{n-1}\textbf{C}_d\textbf{A}^{n-1-i}\textbf{B}\textbf{u}[i]+\textbf{D}_d\textbf{u}[n],  \hspace{3 mm} \forall n\geq 0. \label{general_sol}
\end{align}
In Algorithm 1, it is assumed that $\textbf{x}_0=\textbf{0}$.
 After the source nodes send basis vectors, as described in step 1 in the algorithm, every sink node has the matrices given in the following lemma.
 \begin{lemma}
 For a network described by the state equations (\ref{state-space-1]})-(\ref{state-space-2]}) with $\textbf{x}_0=\textbf{0}$, let the input sequence $\textbf{u}[n]$ be given by
 \begin{align}
 \textbf{u}_{i}[n]=\begin{cases}
             \textbf{e}_i,&  n=0\\
             \textbf{0},& 1\leq n \leq 2N
         \end{cases}, \label{equation_lemma1_input_squence}
 \end{align}
 where $\textbf{e}_i$ is a basis vector of the form
 \begin{align}
 \textbf{e}_i=\left(a_0,a_1,...,a_m\right)^T, \hspace{3 mm} a_k=\begin{cases}
               1,&  k=i\\
               0,& k\neq i
           \end{cases}. \label{basis_vectors}
 \end{align}
   The output sequence in that case will be
 \begin{align}
 \textbf{y}_{d,i}[n]=\begin{cases}
              \textbf{D}_d\cdot \textbf{e}_i,&  n=0\\
              \textbf{C}_d\textbf{A}^{n-1}\textbf{B}\cdot\textbf{e}_i,& 1\leq n \leq 2N
          \end{cases}.
 \end{align}
 Moreover, if one combines the output vectors into matrices $\textbf{M}_d[n]=\left[\textbf{y}_{d,1}[n],...,\textbf{y}_{d,m}[n]\right]$,
 then the corresponding matrices will be
 \begin{align}
 \textbf{M}_d[n]&=\textbf{C}_d\textbf{A}^{n-1}\textbf{B}, \hspace{5 mm} \forall  1\leq n \leq 2N, \label{Markov} \\
 \textbf{M}_d[0]&=\textbf{D}_d. \nonumber
 \end{align} \label{lemma_markov_parameters}
 \end{lemma}
 \begin{proof}
 The proof for this lemma follows immediately by substituting the input sequence from (\ref{equation_lemma1_input_squence}) into the general solution given in (\ref{general_sol}) and using the fact that the initial state $\textbf{x}_0$ is zero.
  \end{proof}
 
 The above matrices $\left\{\textbf{M}_d[n]\right\}$ are usually called the \textit{Markov parameters} of an LTI system.
To continue, we state the Cayley-Hamilton\cite[p.~284]{CITE_Linear_Algebra} Theorem, since it plays an important role in our derivation. 
\begin{theorem}[Cayley-Hamilton Theorem]
For a given $n\times n$ matrix $\textbf{A}$ over the field $\mathbb{F}$, let $P_{\textbf{A}}(t)=\det(t\textbf{I}-\textbf{A})$ be the characteristic polynomial of $\textbf{A}$. Let $\{a_k\}_{k=0}^{n-1}$ be the coefficients of $P_{\textbf{A}}(t)$, so that it can be represented as
\begin{align}
P_{\textbf{A}}(t)=t^n+\sum_{k=0}^{n-1} a_kt^k.
\end{align}
 Then the following holds:
\begin{align}
P_{\textbf{A}}(\textbf{A})=\textbf{A}^n+\sum_{k=0}^{n-1} a_k\textbf{A}^k=\textbf{O},
\end{align}
where $\textbf{O}$ is the zero $n\times n$ matrix.
 \label{theorem_Cayley_Hamilton}
\end{theorem}

We now look for a non zero polynomial $P_d(t)=\sum_{k=0}^{N} \alpha_{d,k}t^k$ that will satisfy 
\begin{align}
\textbf{C}_dP_d(\textbf{A})\textbf{A}^\tau \textbf{B}=\textbf{O}, \hspace{5 mm} \forall \tau\in\mathbb{N}. \label{Inf_Mamal}
\end{align}
We will show later that this polynomial is used in the difference equation (\ref{diffeq}), which is needed for decoding the transmitted symbols.
The set of linear equations given in (\ref{Inf_Mamal}) has an infinite number of equations; using the Cayley-Hamilton Theorem it has at least one solution, where $P_d(t)$ is the characteristic polynomial of $\textbf{A}$.
It is interesting to note that to find $P_d(t)$, we do not need all of these equations since they are linearly dependent. In fact, we have the following lemma that tells us how many equations we need.
\begin{lemma} \label{lemma_infinite-finite_mamal}
If a polynomial $P_d(t)$ satisfies 
\begin{align}
\textbf{C}_dP_d(\textbf{A})\textbf{A}^\tau \textbf{B}=\textbf{O}, \hspace{5 mm} \forall \tau\in\{0,1,...,N-1\}, \label{Fin_Mamal}
\end{align}
 where $\textbf{A}$ is a square $N\times N$ matrix, then it also satisfies (\ref{Inf_Mamal}).
 \label{lemma_Infinite_to_Finite_equations}
\end{lemma}

The proofs for this lemma and for those of all the other  theorems are given in Appendix B.
If we denote the unknown polynomial by $P_d(t)=\sum_{k=0}^{N}\alpha_{d,k}z^k$, then the set of linear equations given in (\ref{Fin_Mamal}) can be written as
\begin{align}
\forall \tau\in\{1,...,N\}: \hspace{2 mm} \textbf{O}&=\textbf{C}_dP_d(\textbf{A})\textbf{A}^{\tau-1} \textbf{B} \\
&=\sum_{k=0}^{N}\alpha_{d,k}\textbf{C}_d\textbf{A}^{k+\tau-1}\textbf{B} \\
&\stackrel{(a)}{=}\sum_{k=0}^{N}\alpha_{d,k}\textbf{M}_d[k+\tau], \label{equation_derivation1_explanation_FinMamal}
\end{align}
where (a) is obtained from Lemma \ref{lemma_markov_parameters}. 
We therefore see that the set of linear equations solved in (\ref{mamal}) is the same set as in (\ref{Fin_Mamal}). 
The next theorem describes the relationship between the sent and received symbols in the network.
\begin{theorem} \label{theorem_decoding}
For a given network and a sink node $d$, let $P_d(z)=\sum_{k=0}^N \alpha_{d,k} z^k$ and $\textbf{G}_d(z)$ be the polynomial and the matrix defined in (\ref{Polynom})-(\ref{TransMatrix}). Then (\ref{diffeq}) holds. 
Furthermore, it is possible to decode  $\textbf{u}$ from $\textbf{y}_d$ if and only if the matrix $\textbf{G}_d(z)$ is of full column rank over the polynomial ring $\mathbb{F}[z]$. 
\end{theorem}

Theorem \ref{theorem_decoding} gives us a way to decode the sent symbols, and it assures us that if the set of linear equations in (\ref{diffeq}) does not have a unique solution, then there is no way for us to find $\textbf{u}$ from $\textbf{y}_d$, even if we know the network topology and the local encoding kernels.

\section{Derivation of Algorithm 2}
We are interested, again, in a difference equation between $\textbf{u}$ and $\textbf{y}_d$, as given in (\ref{diffeq}), that does not depend on $\textbf{x}_0$. Let $\{\textbf{e}_k\}_{k=1}^m$ be the standard basis for the vector space $\mathbb{F}^m$, namely, the elements of the vector $\textbf{e}_k$ are zeros except for the $k$'th element which is equal to one. As described in step 1 of Algorithm 2,  the input sequence $\textbf{u}[n]$ is given by
\begin{align}
\textbf{u}[n]=\sum_{k=1}^{m} \textbf{e}_k1_{\{n=(2N+1)k\}}, \label{input_sequence_no_reset}
\end{align}
where $1_{\{\cdot\}}$ is the indicator function
\begin{align}
1_{\Omega}=\begin{cases}
          1,& \text{statement $\Omega$ is true} \\
          0,& \text{otherwise}
      \end{cases}. \nonumber
\end{align}
We can get the output sequence
if we substitute the above input sequence into the general solution (\ref{general_sol}) of the network's state equations.
The output sequence $\textbf{y}_d[n]$ for $n\geq 0$ will be
\begin{align}
\textbf{y}_d[n]&=\textbf{C}_d\textbf{A}^{n}\textbf{x}_0  +\sum_{k=1}^{m}\sum_{i=0}^{n-1}\textbf{C}_d\textbf{A}^{n-1-i}\textbf{B}\textbf{e}_k1_{\{i=(2N+1)k\}}+\textbf{D}_d\sum_{k=1}^{m}\textbf{e}_k1_{\{n=(2N+1)k\}} \\
&=\textbf{C}_d\textbf{A}^{n}\textbf{x}_0  +\sum_{k=1}^{\min\left\{m,\floor{\frac{n-1}{2N+1}}\right\}}\textbf{C}_d\textbf{A}^{n-1-(2N+1)k}\textbf{B}\textbf{e}_k+\textbf{D}_d\textbf{e}_{n/(2N+1)}1_{\{n/(2N+1)\in\mathbb{N}\}}. \label{output_sequence_no_reset}
\end{align}
We look for a non zero polynomial $P_d(t)$ that will satisfy 
\begin{align}
\textbf{C}_dP_d(\textbf{A})\textbf{A}^\tau\textbf{B}&=\textbf{O}, \hspace{5 mm} \forall \tau\geq 0, \label{InfMamalNew1}\\
\textbf{C}_dP_d(\textbf{A})\textbf{A}^{\tau+1}\textbf{x}_0&=\textbf{0}, \hspace{5 mm} \forall \tau\geq 0, \label{InfMamalNew2}
\end{align}
where $\textbf{O}$ is the $l_d\times m$ zero matrix and $\textbf{0}$ is the zero column vector of dimension $l_d$.
This polynomial is used in the difference equation (\ref{diffeq}), which is needed for decoding the transmitted symbols.
We can limit ourselves to $\tau\in\{0,...,N-1\}$, as can be seen by the next lemma.
\begin{lemma}
Let $P_d(t)$ be a polynomial in $t$ and $\textbf{A}$ is a square $N\times N$ matrix. If either of the equations (\ref{InfMamalNew1}) or (\ref{InfMamalNew2}) hold for $\tau\in\{0,...,N-1\}$, then it also holds for all $\tau\geq N$. \label{lemma_infinite-finite_mamal2}
\end{lemma}

The proof is similar to the proof for Lemma \ref{lemma_infinite-finite_mamal} and is therefore omitted. A method for finding such a polynomial from the received sequence $\left\{\textbf{y}_d[n]\right\}_{1\leq n \textless (m+1)(2N+1)}$ is given in the next theorem.
\begin{theorem} \label{theorem_alg2_polynomial_mamal}
The polynomial $P_d(t)=\sum_{k=0}^{N}\alpha_{d,k}t^k$ satisfies (\ref{InfMamalNew1})-(\ref{InfMamalNew2}) if and only if its coefficients are a solution of (\ref{mamalNew}).      
\end{theorem}
Using the Cayley-Hamilton Theorem, there is at least one polynomial that satisfies (\ref{InfMamalNew1})-(\ref{InfMamalNew2}), which is the characteristic polynomial of $\textbf{A}$, so (\ref{mamalNew}) has at least one solution. After finding a polynomial $P_d(t)$, we can construct a difference equation  for $\textbf{u}$ and $\textbf{y}_d$ that does not depend on the initial state $\textbf{x}_0$. The equation will hold for any time after the initialization process  finishes, even without resetting the state vector.
\begin{theorem}
Let $P_d(z)$ and $\textbf{G}_d(z)$ be the polynomial and the matrix defined in (\ref{Polynom2})-(\ref{TransMatrixNew1})
 Then the following difference equation holds:
\begin{align}
\left(P_d(z)\textbf{y}_d\right)[n] = \left(\textbf{G}_d(z)\textbf{u}\right)[n], \hspace{3 mm} \forall n\geq 1. \label{equation_theorem_main2}
\end{align} \label{theorem_main2}
\end{theorem} 

Equation (\ref{equation_theorem_main2}) itself is not enough for a decoding algorithm since it holds only for $n\geq 1$. In order to decode we need the first $N$ values of $\textbf{u}$: $\textbf{u}[0],...,\textbf{u}[N]$ to be known a priori to the sink nodes. Note that if we apply Algorithm 2 these values are zeros and, hence, are known a priori. We define:
\begin{align}
\textbf{q}[n]=\begin{cases}
          \left(\textbf{G}_d(z)\textbf{u}\right)[n],& n\leq 0 \\
          P_d(z)\textbf{y}_d(z),& n\geq 1
      \end{cases}.
\end{align}
Note that $\left(\textbf{G}_d(z)\textbf{u}\right)[n]$ can be calculated for $n\leq 0$ if we know  $\textbf{u}[0],...,\textbf{u}[N]$, since
\begin{align}
\left(\textbf{G}_d(z)\textbf{u}\right)[n]=\sum_{k=0}^{N}\textbf{G}_d[k]\textbf{u}[n+k].
\end{align}
Once we have the sequence $\textbf{q}$, we note that it satisfies
\begin{align}
\textbf{q}[n]=\left(\textbf{G}_d(z)\textbf{u}\right)[n], \hspace{3 mm} \forall n\in \mathbb{Z},
\end{align}
so we can use (\ref{diffeq_solved}-\ref{diffeq_solved2}) to find $\textbf{u}[n]$ (if $\textbf{G}_d(z)$ is of full column rank).

\section{Derivation of Algorithm 3}
 A direct consequence of Theorem \ref{theorem_decoding} is the fact that we can find achievable rates for every source node from the matrices $\left\{\textbf{G}_d(z)\right\}_{d\in\mathcal{D}}$. If the transmission rates are not achievable with the given local encoding kernels, then one cannot decode the input sequence $\textbf{u}$ from the output $\textbf{y}_d$.
 This result is stated in the following theorem.
 \begin{theorem}
 For a given network and a sink node $d$, let
 \begin{align}
 	P_d(z)\textbf{y}_d&=\textbf{G}_d(z)\textbf{u}=\sum_{s\in\mathcal{S}} \textbf{G}_{d,s}(z)\textbf{u}_s. \label{diffeq_many_sources}
  \end{align}
  describe the relationship between the input sequence $\textbf{u}$ and the output sequence $\textbf{y}_d$ that was found in Algorithm 1 or 2. 
 For every source node $s\in\mathcal{S}$, let $R_s$ be the transmission rate of $s$ that was set before the initialization algorithm was started.
 Then the rates $(R_s')_{s\in\mathcal{S}}$ are achievable for the sink node $d$ with the current local encoding kernels if and only if for every source node $s\in\mathcal{S}$ there exist $R_s'$ linearly independent column vectors $\textbf{v}_{s,1},...,\textbf{v}_{s,R_s}$ from the columns of the matrix $\textbf{G}_{d,s}(z)$, such that  $\cup_{s\in\mathcal{S}}\cup_{k=1}^{R_s'} \{\textbf{v}_{s,k}\}$ is a set of linearly independent vectors  over the polynomial ring $\mathbb{F}[z]$. \label{theorem_alg4}
 \end{theorem} 
\section{Conclusions}
The use of CNC schemes  requires one to choose local encoding kernels at the relay nodes 
that would allow the sink nodes to decode the transmitted symbols. The coefficients can be chosen randomly to simplify the network code
construction, but this would require the sink nodes to know the transfer function of the network. The algorithms we presented allow the sink nodes
to find a difference equation that enables decoding the transmitted from the received symbols without learning the exact topology of the network
and the chosen local encoding kernels. The capacity region can also be found from the obtained difference equation. 
The algorithms require the source nodes to transmit basis vectors and the sink nodes to solve a set of linear equations. 
Both the amount of transmissions every source node needs to perform and the number of linear equations every sink node needs to solve grow linearly with the number of edges and hence, the algorithms are considered efficient computationally.

\begin{appendices}
\section{Examples}
\begin{example} \label{example_alg1}
Consider the network shown in Fig.\ref{fig_network_algorithm}, with two source nodes $s_1,s_2$, one sink node $d$ and three relay nodes. The field on which the network operates is $\mathbb{F}_{2^8}$ with the primitive polynomial $t^8+t^4+t^3+t^2+1$ used to define the field.
The elements of the field $\mathbb{F}_{2^8}$ are polynomials of the form:
\begin{align}
\sum_{k=0}^{7}a_k t^k, \hspace{3 mm} \forall i: a_i\in\left\{0,1\right\}.  \label{finite_field_scalar}
\end{align}
For simplicity, we use an integer representation for every scalar from the field, such that every scalar is represented by a number between $0$ and $255$ whose binary representation $(a_7,a_6,...,a_0)$ is given by the elements $a_i$ from (\ref{finite_field_scalar}).
 
 All local encoding kernels were generated randomly and  are given by 
\begin{align}
x_1[n+1]&=37x_6[n]+108x_3[n],  \label{exm_local_encoding_kernels_start} \\
x_2[n+1]&=234x_1[n]+203x_8[n],  \\
x_3[n+1]&=245x_7[n]+168x_2[n],  \\
x_4[n+1]&=10x_1[n]+217x_8[n],  \\
x_5[n+1]&=239x_7[n]+174x_2[n],  \\
x_6[n+1]&=194u_{s_1,1}[n]+190u_{s_1,2}[n],  \\
x_7[n+1]&=101u_{s_1,1}[n]+168u_{s_1,2}[n],  \\
x_8[n+1]&=44u_{s_2,1}[n].
\label{exm_local_encoding_kernels_stop}
\end{align}

\begin{figure}[h]{
  \psfrag{s}[][][1]{$s_1$}
  \psfrag{t}[][][1]{$s_2$}
  \psfrag{d}[][][1]{$d$}
  \psfrag{x1}[][][1]{$x_1[n]$}
  \psfrag{x2}[][][1]{$x_2[n]$}
  \psfrag{x3}[][][1]{$x_3[n]$}
  \psfrag{x4}[][][1]{$x_4[n]$}
  \psfrag{x5}[][][1]{$x_5[n]$}
  \psfrag{u1}[][][1]{$x_6[n]$}
  \psfrag{u2}[][][1]{$x_7[n]$}
  \psfrag{u3}[][][1]{$x_8[n]$}
   \centerline{\includegraphics[width=12cm]{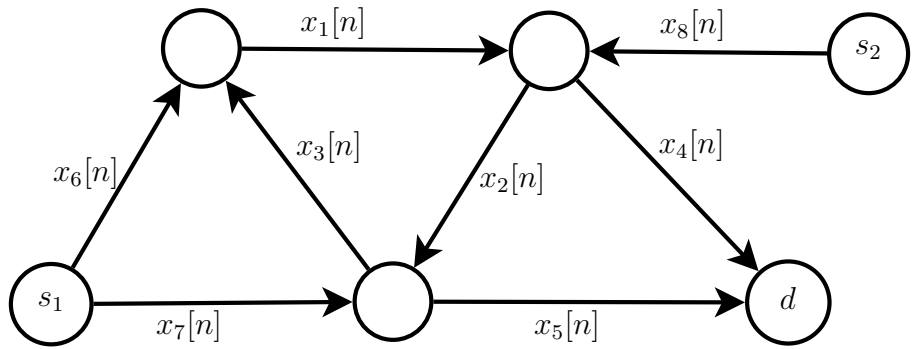}}
  \caption{Network with 2 source nodes, one sink node and 3 relay nodes.} \label{fig_network_algorithm}
  }\end{figure}
  All nodes know only the following facts:
  \begin{itemize}
  \item The source nodes list is $\mathcal{S}=\{s_1,s_2\}$ and the sink node is $d$.
  \item The network has not more than 8 edges ($N=8$).
  \item The number of output links for every source node: $|Out(s_1)|=2$ and $|Out(s_2)|=1$.
  \end{itemize}
  Note that even though the rates $(R_{s_1},R_{s_2})=(2,1)$ are not achievable (they do not satisfy the Min-Cut Max-Flow condition), we assume for now that this information is not known a priori. If it was, we could set the rates to $(R_{s_1},R_{s_2})=(1,1)$ (by setting $u_{s_1,2}[n]=0$) or to $(R_{s_1},R_{s_2})=(2,0)$ (by setting $u_{s_2,1}[n]=0$), since these rates are achievable. 
  In that case the whole initialization process would take 34 time units.. 
  The initialization process begins when $s_1$ and $s_2$ send the following sequences:
  \begin{align}
  u_{s_1,1}[n]&=1,0,0,...,0 \hspace{3 mm} \forall \hspace{1 mm} 0\leq n \leq 16,  \label{equation_first_sequence}\\
  u_{s_1,2}[n]&=0,0,0,...,0 \hspace{3 mm} \forall \hspace{1 mm} 0\leq n \leq 16, \nonumber \\
  u_{s_2,1}[n]&=0,0,0,...,0 \hspace{3 mm} \forall \hspace{1 mm} 0\leq n \leq 16. \nonumber
  \end{align}
 After $n=16$, all the incoming symbols are cleared, $n$ is set to zero, and the following sequences are sent:
    \begin{align}
    u_{s_1,1}[n]&=0,0,0,...,0 \hspace{3 mm} \forall \hspace{1 mm} 0\leq n \leq 16,  \label{equation_second_sequence}\\
    u_{s_1,2}[n]&=1,0,0,...,0 \hspace{3 mm} \forall \hspace{1 mm} 0\leq n \leq 16, \nonumber \\
    u_{s_2,1}[n]&=0,0,0,...,0 \hspace{3 mm} \forall \hspace{1 mm} 0\leq n \leq 16. \nonumber
    \end{align}
    Again, after $n=16$ the network is cleared, and the sent sequences are:
       \begin{align}
       u_{s_1,1}[n]&=0,0,0,...,0 \hspace{3 mm} \forall \hspace{1 mm} 0\leq n \leq 16,  \label{equation_third_sequence}\\
       u_{s_1,2}[n]&=0,0,0,...,0 \hspace{3 mm} \forall \hspace{1 mm} 0\leq n \leq 16, \nonumber \\
       u_{s_2,1}[n]&=1,0,0,...,0 \hspace{3 mm} \forall \hspace{1 mm} 0\leq n \leq 16. \nonumber
       \end{align}
  Meanwhile, the output sequence $\textbf{y}_d=({x}_4,{x}_5)^T$ received by $d$ is given by
    \begin{align}
    \textbf{M}_d[n]&=\left[\textbf{y}_d^{s_1,1}[n], \textbf{y}_d^{s_1,2}[n], \textbf{y}_d^{s_2,1}[n]\right] \hspace{5 mm} \forall n\in\left\{1,...,16\right\} \label{equation_first_received} \\
    &=\begin{bmatrix}
           0 & 0 & 0 \\
           0 & 0 & 0
         \end{bmatrix},
      \begin{bmatrix}
               0 & 0 & 231 \\
               157 & 13 & 0
             \end{bmatrix},
      \begin{bmatrix}
               57 & 73 & 0 \\
               0 & 0 & 228
             \end{bmatrix},
      \begin{bmatrix}
               113 & 63 & 0 \\
               185 & 105 & 0
             \end{bmatrix},
       \begin{bmatrix}
                0 & 0 & 228 \\
                1 & 101 & 0
              \end{bmatrix},..\nonumber.
    \end{align}
  and we have
  \begin{align}
  \textbf{M}_d[n+3]=209\textbf{M}_d[n] \hspace{5 mm} \forall \hspace{1 mm} 3\leq n \leq 13. \label{mahzoriyut}
  \end{align}
  We now solve the set of linear equations given in (\ref{mamal}). We look for some non trivial solution.  
  We can take, for example,
  \begin{align}
  P(z)&=\sum_{k=0}^8 \alpha_{d,k}z^k \\
  &=209z^2+z^5. \label{example_P(z)}
  \end{align}
  This is, indeed, a solution of (\ref{mamal}), as can be seen from (\ref{mahzoriyut}). The matrix $\textbf{G}_d(z)$ given by this solution is:
 \begin{align}
   &\left(209\textbf{M}_d[2]+\textbf{M}_d[5]\right) +\left(209\textbf{M}_d[1]+\textbf{M}_d[4]\right)z+\left(\textbf{M}_d[3]\right)z^2+\left(\textbf{M}_d[2]\right)z^3+\left(\textbf{M}_d[1]\right)z^4= \nonumber \\
      &=\begin{bmatrix}
          113z+57z^2 & 63z+73z^2 & 84+231z^3 \\
                    24+185z+157z^3 & 17+105z+13z^3 & 228z^2
                  \end{bmatrix}.
    \label{example_G(z)}
   \end{align}
   If we are interested in finding achievable rates from the matrix $\textbf{G}_d(z)$, we should apply Algorithm 3, as described in Example \ref{Example_Algorithm3}. For now, we set the rates to achievable ones:  $R_{s_1}=R_{s_2}=1$ (by sending $u_{s_1,2}[n]=0$), and we get the following relationship between the input and the output sequences:
   \begin{align}
   (z^5+209\cdot z^2)\textbf{y}_d=\begin{bmatrix}
             113z+57z^2 & 84+231z^3 \\
                       24+185z+157z^3 &  228z^2
                     \end{bmatrix}
                  \begin{pmatrix}
                     u_{s_1,1} \\
                     u_{s_2,1}
                  \end{pmatrix}. \label{example3_newdiffeq}
   \end{align}
   Note that if we had started the initialization process with the rates $R_{s_1}=R_{s_2}=1$, we would have obtained the following transfer matrix $\tilde{\textbf{G}}_d(z)$:
   \begin{align}
   \tilde{\textbf{G}}_d(z)=\begin{bmatrix}
                113z+57z^2 & 84+231z^3 \\
                          24+185z+157z^3 &  228z^2
                        \end{bmatrix}.
   \end{align} 
   We can solve (\ref{example3_newdiffeq}) for $\textbf{u}$ by multiplying both sides of the equation by $\left(42 \cdot adj(\tilde{\textbf{G}}_d(z))\right)$:
   \begin{align} &(105+223 z+152z^3+149z^4+ z^6)
      \begin{pmatrix}
       u_{s_1,1} \\
       u_{s_2,1}
      \end{pmatrix}= \nonumber \\
    &=\begin{bmatrix}
     221z^4+119z^7 & 65z^2+119z^5+9z^8 \\
     30z^2+208z^3+42z^5+112z^6+241z^8 & 42z^3 + 112 z^4 + 203z^6 + 212 z^7
          \end{bmatrix}
      \textbf{y}_d, \label{equation_example_algorithm1}
      \end{align}
      where we used the identity 
      \begin{align}
      42\cdot adj\left(\tilde{\textbf{G}}_d(z)\right)\tilde{\textbf{G}}_d(z)=42 \det(\tilde{\textbf{G}}_d(z))\textbf{I}.
      \end{align}
      We used the factor $42$ to make the coefficient of $z^6$ from the left side of (\ref{equation_example_algorithm1}) equal one. The difference equation for $\textbf{u}[n]$ is:
      \begin{align}
            \begin{pmatrix}
              u_{s_1,1}[n+6] \\
              u_{s_2,1}[n+6]
          \end{pmatrix}
           &=149 \begin{pmatrix}
        u_{s_1,1}[n+4] \\
       	u_{s_2,1}[n+4]
            \end{pmatrix}+152
            \begin{pmatrix}
        u_{s_1,1}[n+3] \\
        u_{s_2,1}[n+3]
            \end{pmatrix}
            +223 \begin{pmatrix}
        u_{s_1,1}[n+1] \\
        u_{s_2,1}[n+1]
            \end{pmatrix}
            +105 \begin{pmatrix}
        u_{s_1,1}[n] \\
        u_{s_2,1}[n]
            \end{pmatrix} \nonumber \\
   &+       \begin{pmatrix}
     221y_{d,1}[n+4]+119y_{d,1}[n+7] \\
     30y_{d,1}[n+2]+208y_{d,1}[n+3]+42y_{d,1}[n+5]+112y_{d,1}[n+6]+241y_{d,1}[n+8]
            \end{pmatrix}    \nonumber \\
    &+       \begin{pmatrix}
    65y_{d,2}[n+2]+119y_{d,2}[n+5]+9y_{d,2}[n+8] \\
    42y_{d,2}[n+3]+112y_{d,2}[n+4]+203y_{d,2}[n+6]+212y_{d,2}[n+7]
            \end{pmatrix},
            \end{align}
with the initial conditions
\begin{align}
\textbf{u}[n]=\textbf{0}, \forall n<0, \\
\textbf{y}_d[n]=\textbf{0}, \forall n<0.
\end{align}

\end{example}
\begin{example} \label{example_alg2}
We again look at the network in Fig.\ref{fig_network_algorithm}, with the field $\mathbb{F}_{2^8}$ and the same local encoding kernels as in the previous example. Now we assume there is an initial non zero state for the network:
\begin{align}
(x_1[0],...,x_8[0])^T=(50,64,157,121,90,212,149,140)^T.
\end{align}
We follow the instructions of Algorithm 2 to get a difference equation for $\textbf{u}$ and $\textbf{y}_d$. As in Example \ref{example_alg1}, we assume that achievable rates are not known yet and, therefore, we set the transmission rates to $R_{s_1}=2$ and $R_{s_2}=1$.
 At first, the source nodes $s_1$ and $s_2$ transmit the following sequences:
\begin{align}
u_{s_1,1}[n]=\begin{cases}
          1,& n=17 \\
          0,& \text{otherwise}
      \end{cases} \forall 0\leq n \textless 68, \nonumber \\
u_{s_1,2}[n]=\begin{cases}
          1,& n=34 \\
          0,& \text{otherwise}
      \end{cases} \forall 0\leq n \textless 68, \nonumber \\
u_{s_2,1}[n]=\begin{cases}
          1,& n=51 \\
          0,& \text{otherwise}
      \end{cases} \forall 0\leq n \textless 68. \nonumber \\
\end{align}
The output sequence $\left\{\textbf{y}_d[n]\right\}_{1\leq n \textless 68}$ is stored at the sink node $d$.
Here are some of the initial and final values of $\left\{\textbf{y}_d[n]\right\}$:
\begin{align}
\textbf{y}_d[n]=&\begin{pmatrix}
164 \\ 96
\end{pmatrix}, \begin{pmatrix}
253 \\ 6
\end{pmatrix}   \begin{pmatrix}
155 \\ 88
\end{pmatrix},...,\begin{pmatrix}
97 \\ 254
\end{pmatrix} \hspace{5 mm} \forall 1\leq n \leq 19, \\
&\begin{pmatrix}
63 \\ 144
\end{pmatrix}, \begin{pmatrix}
18 \\ 101
\end{pmatrix},\begin{pmatrix}
144 \\ 46
\end{pmatrix},...,\begin{pmatrix}
225 \\ 209
\end{pmatrix},\begin{pmatrix}
172 \\ 108
\end{pmatrix}, \hspace{5 mm} \forall 20\leq n \leq 67.
\end{align}
A solution of (\ref{mamalNew}) leads to the same solution as in the previous example and, therefore the same decoding method can be used.
\begin{align}
P_d(z)&=\sum_{j=0}^{N}\alpha_{d,j}z^j  \nonumber \\
&=z^5+209z^2, \nonumber \\
\textbf{G}_d(z)&=\begin{bmatrix}
          113z+57z^2 & 63z+73z^2 & 84+231z^3 \\
                    24+185z+157z^3 & 17+105z+13z^3 & 228z^2
                  \end{bmatrix}. \nonumber
\end{align}
\end{example}
\begin{example} \label{example_alg3}
We return to the network in Fig. \ref{fig_network_algorithm}, with the same field and coefficients as in Example \ref{example_alg1}. After applying Algorithm 1 or 2, we get the polynomial $P_d(z)$ and the transfer matrix $\textbf{G}_d(z)$, as given in (\ref{example_P(z)}) and (\ref{example_G(z)}). We are interested in achievable rates for the sources $s_1,s_2$, so we follow the instructions given in Algorithm 3. We split $\textbf{G}_d(z)$ into two matrices:
   \begin{align}
   \textbf{G}_{d,s_1}(z)=\begin{bmatrix}
113z+57z^2 & 63z+73z^2 \\
 24+185z+157z^3 & 17+105z+13z^3
                  \end{bmatrix}, \hspace{5 mm}
   \textbf{G}_{d,s_2}(z)=
   \begin{bmatrix}
 84+231z^3 \\
 228z^2                        \end{bmatrix}.
   \end{align}
   The rates $(R_{s_1},R_{s_2})=(1,1)$ are achievable, since the vectors
   \begin{align}
  \textbf{v}_1&=(113z+57z^2, 24+185z+157z^3)^T\in \textrm{Columns of }(\textbf{G}_{d,s_1}), \\
      \textbf{v}_2&=(84+231z^3, 228z^2 )^T\in \textrm{Columns of }(\textbf{G}_{d,s_2})
   \end{align}
   are linearly independent over the polynomial ring $\mathbb{F}[z]$. The rates $(R_{s_1},R_{s_2})=(2,0)$ are also achievable, since $\textbf{G}_{d,s_1}$ is of full rank over the polynomial ring $\mathbb{F}[z]$.

\label{Example_Algorithm3}
\end{example}
\section{Proofs}
\begin{proof}[Proof for Lemma \ref{lemma_Infinite_to_Finite_equations}]
A direct consequence of the Cayley-Hamilton Theorem is that for every $N\times N$ matrix $\textbf{A}$, its power $\textbf{A}^\tau$ can be written as a linear combination of $\textbf{I},\textbf{A},\textbf{A}^2,...,\textbf{A}^{N-1}$ for $\tau\geq N$. Therefore, by substituting this into (\ref{Inf_Mamal}) we get for every  $\tau\geq N$:
\begin{align}
\textbf{C}_dP_d(\textbf{A})\textbf{A}^\tau \textbf{B} &= \textbf{C}_dP_d(\textbf{A})\sum_{i=0}^{N-1}\gamma_i \textbf{A}^i \textbf{B} \\
&= \sum_{i=0}^{N-1} \gamma_i \left(\textbf{C}_dP_d(\textbf{A})\textbf{A}^i\textbf{B}\right) \\
&=\textbf{O},
\end{align}
where the last equality holds because $P_d(z)$ satisfies (\ref{Fin_Mamal}).

\end{proof}
\begin{proof}[Proof for Theorem \ref{theorem_decoding}]
We note first that if the network is described by the state-space equations (\ref{state-space-1]})-(\ref{state-space-2]}), then, by Lemmas \ref{lemma_markov_parameters}-\ref{lemma_infinite-finite_mamal}, the polynomial $P_d(z)=\sum_{k=0}^{N} \alpha_{d,k}z^k$ found in Algorithm 1 satisfies (\ref{Inf_Mamal}), and the transfer matrix $\textbf{G}_d(z)$ is given by
\begin{align}
\textbf{G}_d(z)&=\sum_{k=1}^{N}\sum_{j=k}^N \alpha_{d,j} \textbf{M}_d[j-k+1]z^{k-1}+\textbf{M}_d[0]P_d(z) \\
&=\sum_{k=1}^N \left(\sum_{j=k}^N \alpha_{d,j} \textbf{C}_d \textbf{A}^{j-k}\textbf{B}\right)z^{k-1}+P_d(z)\textbf{D}_d, \label{Gd}
\end{align}
where $\left\{\textbf{M}_d[n]\right\}$ are the Markov parameters of the network.
From (\ref{general_sol}), for every $n \geq 0$ we have:
\begin{align}
\textbf{y}_d[n+1]&=\sum_{i=0}^n \textbf{C}_d\textbf{A}^{n-i}\textbf{B}\textbf{u}[i]+\textbf{D}_d\textbf{u}[n+1].
\end{align}
By applying the $P_d(z)$ operator on both sides we get:
\begin{align}
\left(P_d(z)\textbf{y}_d-P_d(z)\textbf{D}_d\textbf{u}\right)[n+1]&=\sum_{j=0}^N \alpha_{d,j}\textbf{y}_d[n+j+1]-\left(P_d(z)\textbf{D}_d\textbf{u}\right)[n+1].
\end{align}
By expanding $\textbf{y}_d[n+j+1]$ and changing the summation order, we get:
\begin{align}
&\left(P_d(z)\textbf{y}_d-P_d(z)\textbf{D}_d\textbf{u}\right)[n+1]= \\
&=\sum_{j=0}^N \sum_{i=0}^{n+j} \alpha_{d,j}\textbf{C}_d\textbf{A}^{n+j-i}\textbf{B}\textbf{u}[i]\\
 &\stackrel{(a)}{=}\left(\sum_{i=0}^{n}\sum_{j=0}^{N}+\sum_{i=n+1}^{n+N}\sum_{j=i-n}^{N}\right)\alpha_{d,j}\textbf{C}_d\textbf{A}^{n+j-i}\textbf{B}\textbf{u}[i]\\
&=\sum_{i=0}^{n}\textbf{C}_d\left(\sum_{j=0}^{N}\alpha_{d,j}\textbf{A}^{j}\right)\textbf{A}^{n-i}\textbf{B}\textbf{u}[i]+\sum_{i=n+1}^{n+N}\left(\sum_{j=i-n}^{N} \alpha_{d,j}\textbf{C}_d \textbf{A}^{n+j-i}\textbf{B}\right)\textbf{u}[i]\\
&\stackrel{(b)}{=}0+\sum_{k=1}^N\left(\sum_{j=k}^N \alpha_{d,j} \textbf{C}_d\textbf{A}^{j-k}\textbf{B}\right)\textbf{u}[n+k] \label{tzimtzum}\\
&=\left(\sum_{k=1}^N\left(\sum_{j=k}^N \alpha_{d,j} \textbf{C}_d\textbf{A}^{j-k}\textbf{B}\right)z^k\textbf{u}\right)[n],
\end{align}
where
\begin{itemize}
\item[(a)] is obtained by changing the summation order,
\item[(b)] follows from the fact that $P_d(z)$ satisfies (\ref{Inf_Mamal}) and by changing a summation variable $k=i-n$.
\end{itemize}
Therefore, we get:
\begin{align}
P_d(z)\textbf{y}_d&=\left(\sum_{k=1}^{N}\left(\sum_{j=k}^N \alpha_{d,j} \textbf{C}_d\textbf{A}^{j-k}\textbf{B}\right)z^{k-1}+P_d(z)\textbf{D}_d\right)\textbf{u}\\
&=\textbf{G}_d(z)\textbf{u}.
\end{align}
Note that since $\textbf{x}_0=\textbf{0}$ and $\textbf{u}[n]$ and $\textbf{y}_d[n]$ vanish for $n<0$, the state equations (\ref{state-space-1]}-\ref{state-space-2]}) hold for all $n\in\mathbb{Z}$ and, therefore, (\ref{diffeq}) also holds for all $n\in\mathbb{Z}$.

We now prove the second part of the theorem that states that it is possible to decode the input sequence from the output  if and only if the matrix $\textbf{G}_d(z)$ is of full column rank over the polynomial ring $\mathbb{F}[z]$.
If it is not, there exists a vector of polynomials $\textbf{v}_d(z)$ such that
\begin{align}
\textbf{G}_d(z)\textbf{v}_d(z)=\textbf{0}.
\end{align}
Denote the maximal degree of the polynomials in $\textbf{v}_d(z)$ by $\delta$.
Let $\textbf{u}_{d,0}[n]$ be the sequence defined by $\textbf{u}_{d,0}[n]=\left(\textbf{v}(z)\psi\right)[n]$, where the sequence $\psi[n]$ is
\begin{align}
\psi[n]=\begin{cases}
1 &\mbox{if } n=\delta \\
0 &\mbox{if } n\neq \delta
\end{cases}.
\end{align}
Note that $\textbf{u}_{d,0}[n]$ vanishes for $n<0$ and, therefore, $\textbf{u}_{d,0}$ is a legal input sequence 
 that will lead to a zero sequence $\textbf{G}_d(z)\textbf{u}_{d,0}$. In view of (\ref{diffeq}), this will lead to a zero  sequence $P_d(z)\textbf{y}_d$. We assumed that $P_d(z)$ is not the zero polynomial and, hence, the output sequence $\textbf{y}_d$ will also vanish, since the relationship between $\textbf{y}_d$ and $P_d(z)\textbf{y}_d$ is injective. In that case, no decoding method will tell if the input sequence was $\textbf{u}_{d,0}[n]$ or a totally zero sequence.

 On the other hand, if $\textbf{G}_d(z)$ is of full column rank then we can show that the input sequence can be decoded from the output sequence. We apply a decoding scheme that is slightly different to the sequential decoder \cite{CITE_Efficient_code_design}.
 We multiply both sides of (\ref{diffeq}) by $adj(\textbf{G}_d(z))$ (we can assume that $\textbf{G}_d(z)$ is a square matrix, since, if not, we can remove some of its linearly dependent rows to make it square) to get 
 \begin{align}
 \textbf{q} &\coloneqq P_d(z)\textbf{y}_d \label{Eq_q} \\
 \textbf{w} &\coloneqq adj(\textbf{G}_d(z))\textbf{q} \label{diffeq_solved} \\
 &\stackrel{(a)}{=}adj(\textbf{G}_d(z))\textbf{G}_d(z)\textbf{u} \nonumber \\
 &\stackrel{(b)}{=}\det(\textbf{G}_d(z))\textbf{u} \nonumber \\
 &=f_d(z)\textbf{u}, \nonumber
 \end{align}
 where
 \begin{itemize}
 \item[(a)] follows from (\ref{diffeq}),
 \item[(b)] follows from the fact that for any square matrix $\textbf{G}$ over the polynomial ring $\mathbb{F}[z]$, the following identity holds:
 \begin{align}
 adj(\textbf{G})\textbf{G}=\det(\textbf{G})\textbf{I}, \nonumber
 \end{align}
 \end{itemize}
 where $\textbf{I}$ denotes the identity matrix.
 The polynomial $f_d(z)\coloneqq \det(\textbf{G}_d(z))$ is a non zero polynomial, since we assumed that $\textbf{G}_d(z)$ is of full column rank.
 \begin{align}
 f_d(z)&=\sum_{i=0}^{k} \alpha_i z^i, \hspace{2 mm} \alpha_k\neq 0. \nonumber
 \end{align}
If we know the sequence $\textbf{y}_d[n]$ we can compute $\textbf{w}[n]$ and from that find $\textbf{u}[n]$:
\begin{align}
\textbf{w}[n]&=\sum_{i=0}^{k} \alpha_i \textbf{u}[n+i],  \label{diffeq_solved2} \\
\textbf{u}[n]&=\left(\alpha_k\right)^{-1}\left(\textbf{w}[n-k]-\sum_{i=0}^{k-1}\alpha_i \textbf{u}[n-k+i]\right).
\end{align}
\end{proof}
\begin{proof}[Proof for Theorem \ref{theorem_alg2_polynomial_mamal}]
We first prove that if $P_d(t)$ satisfies (\ref{InfMamalNew1})-(\ref{InfMamalNew2}), then its coefficients are a solution of (\ref{mamalNew}). We substitute the expression for $\textbf{y}_d[n]$ from (\ref{output_sequence_no_reset}) into (\ref{mamalNew}) to get:
\begin{align}
&\sum_{j=0}^{N} \alpha_{d,j}\textbf{y}_d[j+\tau]= \\
&=\sum_{j=0}^{N} \alpha_{d,j}\textbf{C}_d\textbf{A}^{j+\tau}\textbf{x}_0+\sum_{j=0}^{N}\sum_{k=1}^{\min\left\{m,\floor{\frac{j+\tau-1}{2N+1}}\right\}}\alpha_{d,j}\textbf{C}_d\textbf{A}^{j+\tau-1-(2N+1)k}\textbf{B}\textbf{e}_k \nonumber\\
&+\sum_{j=0}^{N}\sum_{k=1}^{m}\alpha_{d,j}\textbf{D}_d\textbf{e}_k1_{\left\{j+\tau=(2N+1)k     \right\}} \\
&\stackrel{(a)}{=}\textbf{C}_dP_d(\textbf{A})\textbf{A}^\tau \textbf{x}_0+\sum_{j=0}^{N}\sum_{k=1}^{\floor{\tau/(2N+1)}}\alpha_{d,j}\textbf{C}_d\textbf{A}^{j+\tau-1-(2N+1)k}\textbf{B}\textbf{e}_k \nonumber\\
&+\sum_{j=0}^{N}\sum_{k=1}^{m}\alpha_{d,j}\textbf{D}_d\textbf{e}_k1_{\left\{j+\tau=(2N+1)k     \right\}} \\
&\stackrel{(b)}{=}\textbf{C}_dP_d(\textbf{A})\textbf{A}^{\tau}\textbf{x}_0+\sum_{k=1}^{\floor{\tau/(2N+1)}}\textbf{C}_dP_d(\textbf{A})\textbf{A}^{\tau-1-(2N+1)k}\textbf{B}\textbf{e}_k +\textbf{0} \\
&\stackrel{(c)}{=}\textbf{0},
\end{align}
where
\begin{itemize}
\item[(a)] follows from the fact that $0\leq j \leq N$ and that
\begin{align}
\tau=(2N+1)p+\tilde{\tau}, \hspace{3 mm} \text{where} \hspace{3 mm} 0\leq p \leq m, \hspace{2 mm} 1\leq \tilde{\tau} \leq N,
\end{align}
and therefore
\begin{align}
\floor{\frac{j+\tau-1}{2N+1}}=\floor{p+\frac{j+\tilde{\tau}-1}{2N+1}}=p,
\end{align}  
\item[(b)] follows from the fact that $j+\tau$ cannot be a multiple of $(2N+1)$,
\item[(c)] follows from (\ref{InfMamalNew1})-(\ref{InfMamalNew2}).
\end{itemize}
We now prove the veracity of Theorem \ref{theorem_alg2_polynomial_mamal} in the reverse direction. We assume the coefficients $\left\{\alpha_{d,j}\right\}_{j=0}^N$ of the polynomial $P_d(t)=\sum_{j=0}^{N} \alpha_{d,j}t^j$ satisfy (\ref{mamalNew}) and show that $P_d(t)$ satisfies (\ref{InfMamalNew1})-(\ref{InfMamalNew2}). We first note that
\begin{align}
\textbf{y}_d[n]=\textbf{C}_d\textbf{A}^n\textbf{x}_0 \hspace{5 mm} \forall 0\leq n \leq 2N,
\end{align}
so if (\ref{mamalNew}) is satisfied for $\tau\in\{1,...,N\}$, then  $(\ref{InfMamalNew2})$ is also satisfied for $\tau\in\{0,...,N-1\}$ and, hence, for all $\tau\geq 0$ (by Lemma \ref{lemma_infinite-finite_mamal2}). For $\tau\in\left\{(2N+1)+1,...,(2N+1)+N\right\}$, we have:
\begin{align}
\textbf{0}&=\sum_{j=0}^{N}\alpha_{d,j}\textbf{y}_d[j+\tau] \\
&=\sum_{j=0}^{N}\alpha_{d,j}\textbf{C}_d\textbf{A}^{j+\tau-1-(2N+1)}\textbf{B}\textbf{e}_1+\sum_{j=0}^{N}\alpha_{d,j}\textbf{C}_d\textbf{A}^{j+\tau}\textbf{x}_0 \\
&\stackrel{(a)}{=}\textbf{C}_dP_d(\textbf{A})\textbf{A}^{\tau-1-(2N+1)}\textbf{B}\textbf{e}_1+\textbf{0},
\end{align}
where (a) is because (\ref{InfMamalNew2}) holds. In view of Lemma \ref{lemma_infinite-finite_mamal2}, we see that (\ref{InfMamalPartial}) holds for $k=1$.
\begin{align}
\textbf{C}_dP_d(\textbf{A})\textbf{A}^{\tau}\textbf{B}\textbf{e}_k=\textbf{0}, \hspace{5 mm} \forall \tau\geq 0. \label{InfMamalPartial}
\end{align}
By induction on $k$, we can prove that (\ref{InfMamalPartial}) holds for all $k\in\{1,...,m\}$. For all $\tau \in \left\{(2N+1)k+1,...,(2N+1)k+N\right\}$, we have:
\begin{align}
\textbf{0}&=\sum_{j=0}^{N}\alpha_{d,j}\textbf{y}_d[j+\tau] \\
&=\sum_{k'=1}^k\sum_{j=0}^{N}\alpha_{d,j}\textbf{C}_d\textbf{A}^{j+\tau-1-(2N+1)k'}\textbf{B}\textbf{e}_{k'}+\sum_{j=0}^{N}\alpha_{d,j}\textbf{C}_d\textbf{A}^{j+\tau}\textbf{x}_0 \\
&=\textbf{C}_dP_d(\textbf{A})\textbf{A}^{\tau-1-(2N+1)k}\textbf{B}\textbf{e}_k+\sum_{k'=1}^{k-1}\textbf{C}_dP_d(\textbf{A})\textbf{A}^{\tau-1-(2N+1)k'}\textbf{B}\textbf{e}_{k'}+\textbf{0} \\
&\stackrel{(a)}{=}\textbf{C}_dP_d(\textbf{A})\textbf{A}^{\tau-1-(2N+1)k}\textbf{B}\textbf{e}_k+\textbf{0},
\end{align}
where (a) follows from the induction assumption of (\ref{InfMamalPartial}) on $k'\textless k$. In view of Lemma \ref{lemma_infinite-finite_mamal2}, we see that (\ref{InfMamalPartial}) holds for all $0\leq k \leq m$ and, therefore, (\ref{InfMamalNew1}) holds as well.
\end{proof}
\begin{proof}[Proof for Theorem \ref{theorem_main2}]
First we note that the polynomial $P_d(z)=\sum_{k=0}^{N}\alpha_{d,k}z^k$ found in Algorithm 2 satisfies  (\ref{InfMamalNew1})-(\ref{InfMamalNew2}) and the transfer matrix $\textbf{G}_d(z)$ can be described as follows:
\begin{align}
\textbf{G}_d(z)&=\left[\textbf{g}_{d,1}(z),\textbf{g}_{d,2}(z),...,\textbf{g}_{d,m}(z)\right],  \\
  \textbf{g}_{d,i}(z)&=\sum_{k=1}^{N+1} \sum_{j=0}^N \alpha_{d,j}\textbf{y}_{d}[j+(2N+1)i-k+1]z^{k-1}, \hspace{5 mm} \forall i\in\{1,...,m\}, \label{equation_derivation2_g_d_i}
\end{align}
where $\left\{\textbf{y}_d[n]\right\}$ are defined in (\ref{output_sequence_no_reset}).
To prove the theorem, first we find an expression for $\textbf{g}_{d,i}(z)$ by substituting (\ref{output_sequence_no_reset}) into (\ref{equation_derivation2_g_d_i}), and we show that $\textbf{G}_d(z)$ is given by (\ref{Gd}). Then we explain why (\ref{equation_theorem_main2}) holds despite the fact that $\textbf{x}_0\neq \textbf{0}$.
\begin{align}
\textbf{g}_{d,i}(z)=&\sum_{k=1}^{N+1} \sum_{j=0}^N \alpha_{d,j}\textbf{y}_{d}[j+(2N+1)i-k+1]z^{k-1} \nonumber \\
=&\sum_{k=1}^{N+1} \sum_{j=0}^N \alpha_{d,j}\textbf{C}_d\textbf{A}^j\textbf{A}^{(2N+1)i-k+1}\textbf{x}_0z^{k-1}  \nonumber \\
&+\sum_{k=1}^{N+1} \sum_{j=0}^N \sum_{k'=0}^{\min\{m,\floor{\frac{j-k}{2N+1}+i}\}} \alpha_{d,j}\textbf{C}_d\textbf{A}^{j-k+(2N+1)(i-k')}\textbf{B}\textbf{e}_{k'}z^{k-1} \nonumber \\
&+\sum_{k=1}^{N+1} \sum_{j=0}^N \alpha_{d,j}\sum_{k'=1}^{m}\textbf{D}_d\textbf{e}_{k'}1_{\left\{j+(2N+1)i-k+1=(2N+1)k'\right\}}z^{k-1}. \label{g_d,i_murhav}
\end{align}
The first part of (\ref{g_d,i_murhav}) vanishes because of (\ref{InfMamalNew2})
\begin{align}
\sum_{k=1}^{N+1} \sum_{j=0}^N \alpha_{d,j}\textbf{C}_d\textbf{A}^j\textbf{A}^{(2N+1)i-k+1}\textbf{x}_0z^{k-1} &= \sum_{k=1}^{N+1} \textbf{C}_dP_d(\textbf{A})\textbf{A}^{(2N+1)i-k+1}\textbf{x}_0z^{k-1} \nonumber \\
&= \textbf{0}. \nonumber
\end{align}
The second part of (\ref{g_d,i_murhav}) can be written as:
\begin{align}
&\sum_{k=1}^{N+1} \sum_{j=0}^N \sum_{k'=0}^{\min\{m,\floor{\frac{j-k}{2N+1}+i}\}} \alpha_{d,j}\textbf{C}_d\textbf{A}^{j-k+(2N+1)(i-k')}\textbf{B}\textbf{e}_{k'}z^{k-1} \nonumber \\
&\stackrel{(a)}{=} \sum_{k=1}^N \sum_{j=k}^N  \alpha_{d,j}\textbf{C}_d\textbf{A}^{j-k}\textbf{B}\textbf{e}_{i}z^{k-1}
+\sum_{k=1}^{N+1} \sum_{j=0}^N \sum_{k'=0}^{i-1} \alpha_{d,j}\textbf{C}_d\textbf{A}^j\textbf{A}^{(2N+1)(i-k')-k}\textbf{B}\textbf{e}_{k'}z^{k-1} \nonumber \\
&=\sum_{k=1}^N \sum_{j=k}^N  \alpha_{d,j}\textbf{C}_d\textbf{A}^{j-k}\textbf{B}\textbf{e}_{i}z^{k-1}
+\sum_{k=1}^{N+1} \sum_{k'=0}^{i-1} \textbf{C}_dP_d(\textbf{A})\textbf{A}^{(2N+1)(i-k')-k}\textbf{B}\textbf{e}_{k'}z^{k-1} \nonumber \\
 &\stackrel{(b)}{=}\sum_{k=1}^N \sum_{j=k}^N  \alpha_{d,j}\textbf{C}_d\textbf{A}^{j-k}\textbf{B}\textbf{e}_{i}z^{k-1},
\end{align}
where
\begin{itemize}
\item[(a)] is obtained by splitting the summation over $k'$ from $0$ to $i-1$, and for $k'=i$ (only when $k\leq j$).
\item[(b)] holds because $P_d(z)$ satisfies (\ref{InfMamalNew1}).
\end{itemize}
The third part of (\ref{g_d,i_murhav}) is:
\begin{align}
&\sum_{k=1}^{N+1} \sum_{j=0}^N \alpha_{d,j}\sum_{k'=1}^{m}\textbf{D}_d\textbf{e}_{k'}1_{\left\{j+(2N+1)(i-k')-k+1=0\right\}}z^{k-1} \nonumber \\
&=\sum_{j=0}^{N}\alpha_{d,j}z^j\textbf{D}_d\textbf{e}_i \nonumber \\
&=P_d(z)\textbf{D}_d\textbf{e}_i.
\end{align}
By combining all $\textbf{g}_{d,i}$ vectors into a matrix we get that $\textbf{G}_d(z)$ is given by (\ref{Gd}). It was already proved in Theorem \ref{theorem_decoding} that for such a $\textbf{G}_d(z)$ matrix, the difference equation (\ref{diffeq}) holds between $\textbf{u}$ and $\textbf{y}_d$, provided that $\textbf{x}_0=\textbf{0}$. For cases in which the initial state is not zero, the output sequence $\textbf{y}_d[n]$ can be written as a sum of the zero input response $\textbf{y}_{\text{ZIR},d}[n]$ and the zero state response $\textbf{y}_{\text{ZSR},d}[n]$ sequences, where:
\begin{align}
\textbf{y}_{\text{ZIR},d}[n]=\textbf{C}_d\textbf{A}^n\textbf{x}_0,  
\end{align}
and $\textbf{y}_{\text{ZSR},d}[n]$ is the output of the network, as if the initial state were zero. Finally, using Theorem \ref{theorem_decoding} for $\textbf{y}_{\text{ZSR},d}[n]$ and (\ref{InfMamalNew2}) for $\textbf{y}_{\text{ZIR},d}[n]$ we get for all $n\geq 1$:
\begin{align}
\left(P_d(z)\textbf{y}_d\right)[n]&=\left(P_d(z)\textbf{y}_{\text{ZIR},d}\right)[n]+\left(P_d(z)\textbf{y}_{\text{ZSR},d}\right)[n] \\
&=\textbf{C}_dP_d(\textbf{A})\textbf{A}^n\textbf{x}_0+\left(\textbf{G}_d(z)\textbf{u}\right)[n] \\
&=\left(\textbf{G}_d(z)\textbf{u}\right)[n].
\end{align}

\end{proof}
\begin{proof}[Proof for Theorem \ref{theorem_alg4}]
First, we show that if for every $s\in\mathcal{S}$ there are $R_s'$ linearly independent column vectors $\textbf{v}_{s,1},...,\textbf{v}_{s,R_s}$ from the columns of the matrix $\textbf{G}_{d,s}(z)$, such that  $\cup_{s\in\mathcal{S}}\cup_{k=1}^{R_s'} \{\textbf{v}_{s,k}\}$ is a set of linearly independent vectors,  then the rates $(R_s')_{s\in\mathcal{S}}$ are achievable for the sink node $d$ with the current local encoding kernels. Every source node $s$ can transmit its input symbols only on the inputs $u_{s,i}$ that correspond to the vectors $\left\{\textbf{v}_{s,1},...,\textbf{v}_{s,R_s'}\right\}$ and zeros on the other inputs:
\begin{align}
u_{s,i}= \begin{cases}
\text{The zero sequence} &\mbox{if } \text{column $i$ of $\textbf{G}_{d,s}(z)$} \notin \left\{\textbf{v}_{s,1},...,\textbf{v}_{s,R_s'}\right\} \\
\text{A non zero sequence} &\mbox{if } \text{column $i$ of $\textbf{G}_{d,s}(z)$} \in \left\{\textbf{v}_{s,1},...,\textbf{v}_{s,R_s'}\right\}
\end{cases}.
\end{align}
In that case (\ref{diffeq}) can be simplified into
\begin{align}
P_d(z)\textbf{y}_d=\tilde{\textbf{G}}_d(z)\tilde{\textbf{u}}, \label{diffeq_tilde}
\end{align}
where $\tilde{\textbf{u}}[n]$ is the input sequence with all zero inputs  removed and $\tilde{\textbf{G}}_d(z)$ is the matrix 
\begin{align}
\tilde{\textbf{G}}_d(z)=\left[\textbf{v}_{s_1,1},...,\textbf{v}_{s_1,R_{s_1}'} \textbf{v}_{s_2,1},...,\textbf{v}_{s_2,R_{s_2}'},...,\textbf{v}_{s_{|\mathcal{S}|},R_{s_{|\mathcal{S}|}}'} \right].
\end{align}
From the proof of the second part of Theorem \ref{theorem_decoding}, we know that it is possible to decode $\tilde{\textbf{u}}$ from $\textbf{y}_d$ if and only if the matrix $\tilde{\textbf{G}}_{d}(z)$ is of full column rank, i.e. all of its column vectors are linearly independent over $\mathbb{F}[z]$. We assumed that the columns of $\tilde{\textbf{G}}_d(z)$ are linearly independent and, therefore, the rates $(R_s')_{s\in\mathcal{S}}$ are achievable for the sink node $d$.

We now prove the second part of Theorem \ref{theorem_alg4}. We assume that the rates $(R_s')_{s\in\mathcal{S}}$ are achievable for a sink node $d$ and we show that for every $s\in\mathcal{S}$ there are $R_s'$ linearly independent column vectors $\textbf{v}_{s,1},...,\textbf{v}_{s,R_s'}$ from the columns of the matrix $\textbf{G}_{d,s}(z)$, such that the vectors $\cup_{s\in\mathcal{S}}\cup_{k=1}^{R_s'} \{\textbf{v}_{s,k}\}$ are linearly independent. By the definition of achievable rates, we know that every source $s\in\mathcal{S}$ can transmit zeros on $(R_s-R_s')$ of its input sequences, such that $d$ will be able to decode $\textbf{u}$ from $\textbf{y}_d$. If that is the case, the relationship between the input and the output sequences is described by (\ref{diffeq_tilde}),  where $\tilde{\textbf{u}}[n]$ is the input sequence with all zero inputs  removed, and $\tilde{\textbf{G}}_d(z)$ is the matrix $\textbf{G}_d(z)$ with removed column vectors that correspond to the zero input sequences. Since $\tilde{\textbf{u}}$ is decodable, we know that the column vectors of $\tilde{\textbf{G}}_d(z)$ are linearly independent over $\mathbb{F}[z]$. The $R_s'$ column vectors of $\tilde{\textbf{G}}_d(z)$ that correspond to the non zero inputs of a source node $s$ are also column vectors of $\textbf{G}_{d,s}(z)$, since $\textbf{G}_{d,s}(z)$ contains all the column vectors of $\textbf{G}_{d}(z)$ that correspond to the input sequence $\textbf{u}_{s}$. Therefore, we showed that for every source $s$, there are $R_s'$ linearly independent vectors from the columns of $\textbf{G}_{d,s}(z)$, such that all the vectors together are linearly independent. This completes the proof.
\end{proof}
\end{appendices}

\section*{Acknowledgments}
The authors would like to thank Dr. Izchak Lewkowicz for his lectures in Linear Control theory, which were of great help. 

\bibliographystyle{IEEEtran}
\bibliography{IEEEabrv,citat}

\end{document}